\documentclass[11pt]{article}
\pdfoutput=1

\usepackage{microtype}
\usepackage{cite}
\usepackage{mdframed}
\usepackage{csquotes}
\usepackage[dvipsnames,svgnames,x11names]{xcolor}
\usepackage{soul}

\usepackage{listings}
\usepackage{mathrsfs}
\usepackage[mathscr]{euscript}
\usepackage{amsfonts,amsthm,amsmath}
\usepackage{amssymb}
\usepackage{amstext}
\usepackage{graphicx,tikz-cd}
\RequirePackage{fancyhdr}
\usepackage{boxedminipage}
\usepackage[linesnumbered, boxed, algosection,noline]{algorithm2e}
\usepackage{graphics}
\usepackage{pgf}
\usepackage{paralist}
\usepackage{framed}
\usepackage{thm-restate}
\usepackage{comment}
\usepackage{thmtools}
\usepackage{enumerate}
\usepackage{float}
\usepackage{xspace}
\usepackage{todonotes}
\usepackage{graphicx}
\usepackage{caption}
\usepackage{subcaption}
\usepackage{enumitem}
\usepackage{makecell}
\usepackage{thmtools, thm-restate}
 \usepackage{mathdots}
\usepackage{multirow}

\usepackage[T1]{fontenc}

\newcommand{\cmsotwo}{$\mathsf{CMSO}_2$\xspace}

\usepackage[margin=1in,nohead]{geometry}

\newtheorem{theorem}{Theorem}[section]
\newtheorem{lemma}{Lemma}[section]

\newtheorem{observation}{Observation}[section]

\theoremstyle{definition}
\newtheorem{definition}{Definition}[section]

\date{}

\newcommand{\tw}{\mathrm{tw}}
\newcommand{\td}{\mathcal{T}}
\newcommand{\tin}{\mathsf{tree} \textnormal{-} \alpha}
\newcommand{\trlambda}{\mathsf{tree} \textnormal{-} \lambda}
\newcommand{\trmu}{\mathsf{tree} \textnormal{-} \mu}

\renewcommand{\P}{\textup{\textsf{P}}}
\newcommand{\NP}{\textup{\textsf{NP}}}

\newcommand{\parNP}{\textup{\textsf{para-NP}}}
\newcommand{\OO}{\mathcal{O}}
\newcommand{\FPT}{\textup{\textsf{FPT}}}
\newcommand{\W}{\textup{\textsf{W[1]}}}

\usepackage{booktabs}
\usepackage{hyperref}
\hypersetup{colorlinks=true,
            citecolor=green!40!black,
            linkcolor=red!20!black,
            urlcolor=blue!40!black,
            filecolor=cyan!30!black} 

\usepackage[nameinlink,capitalise]{cleveref}

\crefname{observation}{Observation}{Observations}
\Crefname{observation}{Observation}{Observations}

\author{
Cl\'{e}ment Dallard\thanks{
Department of Informatics, University of Fribourg, Switzerland.}
\and
Fedor V. Fomin\thanks{
Department of Informatics, University of Bergen, Norway.}
\and
Petr A. Golovach\addtocounter{footnote}{-1}\footnotemark{}
\and
Tuukka Korhonen\thanks{Department of Computer Science, University of Copenhagen, Denmark.}
\and
Martin Milani{\v c}\thanks{FAMNIT and IAM, University of Primorska, Koper, Slovenia.}
}

\title{Computing Tree Decompositions with Small Independence Number\thanks{The research leading to these results has received funding from the Research Council of Norway via the project BWCA (grant no.\ 314528), by the  Slovenian Research and Innovation Agency (I0-0035, research program P1-0285 and research projects J1-3003, J1-4008, J1-4084, J1-60012, and N1-0370), and by the research program CogniCom (0013103) at the University of Primorska.\\
An extended abstract of this work was presented at the 51st International Colloquium on Automata, Languages, and Programming (ICALP 2024).}
}

\begin{document}

\maketitle

\begin{abstract}
The independence number of a tree decomposition is the maximum of the independence numbers of the subgraphs induced by its bags.
The tree-independence number of a graph is the minimum independence number of a tree decomposition of it.
Several \NP-hard graph problems, like maximum weight independent set, can be solved in time $n^{\OO(k)}$ if the input $n$-vertex graph is given together with a tree decomposition of independence number~$k$.
Yolov, in [SODA 2018], gave an algorithm that, given an $n$-vertex graph $G$ and an integer $k$, in time
$n^{\OO(k^3)}$ either constructs a tree decomposition of $G$ whose independence number is $\OO(k^3)$ or correctly reports that
the tree-independence number of $G$ is larger than~$k$.

In this paper, we first give an algorithm for computing the tree-independence number with a better approximation ratio and running time and then prove that our algorithm is, in some sense, the best one can hope for.
More precisely, our algorithm runs in time $2^{\OO(k^2)} n^{\OO(k)}$ and either outputs a tree decomposition of $G$ with independence number at most $8k$, or determines that the tree-independence number of $G$ is larger than~$k$.
This implies $2^{\OO(k^2)} n^{\OO(k)}$-time algorithms for various problems, like maximum weight independent set, parameterized by the tree-independence number $k$ without needing the decomposition as an input.
Assuming Gap-ETH, an $n^{\Omega(k)}$ factor in the running time is unavoidable for any approximation algorithm for the tree-independence number.

Our second result is that the exact computation of the tree-independence number is \parNP-hard: We show that for every constant $k \ge 4$ it is \NP-complete to decide if a given graph has the tree-independence number at most~$k$.
\end{abstract}

\section{Introduction}

Tree decompositions are among the most popular tools in graph algorithms.
The crucial property of tree decompositions exploited in the majority of dynamic programming algorithms is that each bag of the decomposition can interact with an optimal solution only in a bounded number of vertices. The common measure of a tree decomposition is the width, that is,  the maximum size of a bag in the decomposition (minus 1). The corresponding graph parameter is the treewidth of a graph.
Many problems that are intractable on general graphs can be solved efficiently when the treewidth of a graph is bounded.

However, it is not always the size of a bag that matters. For example, suppose that every bag of the decomposition is a clique, that is, the graph is chordal. Since every independent set intersects each of the clique-bags in at most one vertex, dynamic programming still computes maximum weight independent sets in such graphs in polynomial time even if the bags could be arbitrarily large.
An elegant approach to capturing such properties of tree decompositions 
is the notion of the tree-independence number of a graph.
The \emph{independence number} of a tree decomposition is the maximum of the independence numbers (that is, the maximum size of an independent set) of the subgraphs induced by its bags.
The \emph{tree-independence number} of a graph $G$, denoted by  $\tin(G)$,  is the minimum independence number of a tree decomposition of~$G$. In particular, the tree-independence number of a chordal graph is at most one, and for any graph $G$, the value $\tin(G)$ does not exceed the treewidth of $G$ (plus 1) or the independence number $\alpha(G)$ of~$G$.

The family of graph classes with bounded tree-independence number forms a significant generalization of graph classes with bounded treewidth.
It also contains dense graph classes, including: graph classes with bounded independence number;
classes of intersection graphs of connected subgraphs of graphs with bounded treewidth, studied by Bodlaender et al.~\cite{bodlaender1998}, which in particular include classes of \emph{{$H$}-graphs}, that is, intersection graphs of connected subgraphs of a subdivision of a fixed multigraph $H$, introduced in 1992 by B\'{\i}r\'{o} et al.~\cite{MR1172354} and studied more recently in a number of papers~\cite{MR4249058,MR4141534,MR4332111}; classes of graphs in which all minimal separators have bounded size, studied by Skodinis in 1999~\cite{Skodinis99}; and, more generally, classes of graphs in which all minimal separators induce subgraphs with bounded independence number, studied  by Dallard et al.~\cite{DallardMS24}.

\paragraph{Our results.}
Yolov \cite{Yolov18} gave an algorithm that for a given $n$-vertex graph $G$ and integer $k$, in time
$n^{\OO(k^3)}$ either constructs a tree decomposition of $G$ whose independence number is $\OO(k^3)$ or correctly reports that the tree-independence number of $G$ is larger than~$k$.
Our first main result is the following improvement over Yolov's algorithm.

\begin{theorem}
\label{the:main_alg}
There is an algorithm that, given an $n$-vertex graph $G$ and an integer $k$, in time $2^{\OO(k^2)} n^{\OO(k)}$ either outputs a tree decomposition of $G$ with independence number at most $8k$, or concludes that the tree-independence number of $G$ is larger than~$k$.
\end{theorem}

The performance of the algorithm from  \Cref{the:main_alg}  (the running time and the need of approximation) is in some sense optimal, for the following reasons.
First,  by a simple reduction that given an $n$-vertex graph $G$ produces a $2n$-vertex graph $G'$ with $\tin(G') = \alpha(G)$ (see~\cite{DallardMS24}, as well as \Cref{lem:inappr}), all (parameterized) hardness results for the independence number (or clique) translate into hardness results for computing the tree-independence number.
In particular, from the result of Lin~\cite{Lin21} it follows that a constant-factor approximation of the tree-independence number is \W-hard, and from the result of Chalermsook at al.~\cite{DBLP:journals/siamcomp/ChalermsookCKLM20} it follows that assuming Gap-ETH\footnote{Gap-ETH states that for some constant $\epsilon > 0$, distinguishing between a satisfiable 3-\textsf{SAT} formula and one that is not even $(1-\epsilon)$-satisfiable requires exponential time (see~\cite{DBLP:conf/icalp/ManurangsiR17,DBLP:journals/eccc/Dinur16}).}, there is no $f(k) \cdot n^{o(k)}$-time $g(k)$-approximation algorithm for the tree-independence number for any computable functions $f$ and~$g$.
In particular, the time complexity obtained in \cref{the:main_alg} is optimal in the sense that an $n^{\Omega(k)}$ factor is unavoidable assuming Gap-ETH, even if the approximation ratio would be drastically weakened.

The above arguments do not exclude the possibility of  \emph{exact} computation of tree-independence number in time $n^{f(k)}$ for some function~$f$.
The computational complexity  of recognizing graphs with the tree-independence numbers at most $k$ for a fixed integer $k \geq 2$ was asked as an open problem by Dallard et al.~\cite[Question~8.3]{DallardMS24}.\footnote{There is a linear-time algorithm for deciding if a given graph has the tree-independence number at most $1$, because such graphs are exactly the chordal graphs, see, e.g., \cite{Golumbic80}.}
While for values $k=2$ and $k=3$ the question remains open,  our next result resolves it for any constant $k\geq 4$.

\begin{restatable}{theorem}{npcmain}
\label{the:main_nphard}
For every constant $k\geq 4$, it is \NP-complete to decide whether $\tin(G)\leq k$ for a given graph~$G$.
\end{restatable}

 Let us observe that  \cref{the:main_nphard} implies also that, assuming $\P \neq \NP$, there is no $n^{f(k)}$-time approximation algorithm for the tree-independence number with approximation ratio less than~$5/4$.

We supplement our main results with a second \NP-completeness proof for a problem closely related to computing the tree-independence number and the algorithm of \Cref{the:main_alg}.
We consider the problem where we are given a graph $G$, two non-adjacent vertices $u$ and $v$, and an integer $k$, and the task is to decide if $u$ and $v$ can be separated by removing a set of vertices that induces a subgraph with independence number at most~$k$.
We show in \Cref{thm:separator} that this problem is \NP-complete for any fixed integer $k \ge 3$.
This hardness result is motivated by the fact that the algorithm of \Cref{the:main_alg} finds separators with bounded independence number as a subroutine.
While for the algorithm of \Cref{the:main_alg}, we design a $2^{\OO(k^2)} n^{\OO(k)}$-time $2$-approximation algorithm for a generalization of this problem (assuming that a tree decomposition of independence number $\OO(k)$ is given), the proof of \Cref{thm:separator} shows that this step of the algorithm cannot be turned into an exact algorithm (in our reduction, we can construct, along the graph $G$, a tree decomposition of independence number $\OO(k)$ of $G$).

\paragraph{Previous work and applications of  \Cref{the:main_alg}.}
The notion of the tree-independence number is very natural and it is not surprising that it was introduced independently by several researchers~\cite{DallardMS24,Yolov18}.
Yolov, in \cite{Yolov18}, introduced a new width measure called
\emph{minor-matching hypertree-width}, $\trmu$\footnote{Minor-matching hypertree-width is defined for hypergraphs, but algorithms for computing decompositions for it are only known for graphs.}.
He proved that a number of problems, including \textsc{Maximum Independent Set}, \textsc{$k$-Colouring}, and \textsc{Graph Homomorphism}, permit polynomial-time solutions for graphs with bounded minor-matching hypertree width.
Furthermore, Yolov showed that the minor-matching hypertree-width of a graph is equal to the tree-independence number of the square of its line graph, that is, $\trmu(G) = \tin(L(G)^2)$ holds for all graphs $G$, where $L(G)^2$ is the graph whose vertices are the edges of $G$, with two distinct edges adjacent if and only if they have nonempty intersection or there is a third edge intersecting both.
Moreover, a tree decomposition of $L(G)^2$ with independence number at most $k$ can be turned into a tree decomposition of $G$ with minor-matching hypertree-width at most~$k$.
Then, Yolov gave an $n^{\OO(k^3)}$-time $\OO(k^2)$-approximation algorithm computing the tree-independence number of an $n$-vertex graph, implying also the same bounds for computing the minor-matching hypertree-width of a graph.
\Cref{the:main_alg} improves the running time and approximation ratio of Yolov's algorithm. Pipelined with Yolov's reduction,  \Cref{the:main_alg} also implies an $8$-approximation of minor-matching hypertree-width of graphs in time $2^{\OO(k^2)} n^{\OO(k)}$.

\Cref{the:main_nphard} implies also the \NP-hardness of deciding whether $\trmu(G) \le k$ for every constant $k \ge 4$ because a simple reduction that attaches a pendant vertex to every vertex of a graph $G$ produces  a graph $G'$ such that $\trmu(G') = \tin(G)$ (see~\cite{Yolov18}).

The tree-independence number of a graph was introduced independently by Yolov~\cite{Yolov18} and Dallard et al.~\cite{DallardMS24}.\footnote{Yolov called it \emph{$\alpha$-treewidth} in~\cite{Yolov18}.}
The original motivation for Dallard et al.~\cite{DallardMS24} stems from structural graph theory.
In 2020, Dallard et al.~\cite{DMS-WG2020,dallard2021treewidth} initiated a systematic study of $(\tw,\omega)$-bounded graph classes, that is, hereditary graph classes in which the treewidth can only be large due to the presence of a large clique.
While $(\tw,\omega)$-bounded graph classes are known to possess some good algorithmic properties related to clique and coloring problems (see~\cite{DBLP:journals/endm/ChaplickZ17,MR4332111,DMS-WG2020,dallard2021treewidth,DallardMS24}), the extent to which this property has useful algorithmic implications for problems related to independent sets is an open problem.
The connection with the tree-independence number follows from Ramsey's theorem, which implies that graph classes with bounded tree-independence number are $(\tw,\omega)$-bounded with a polynomial binding function (see~\cite{DallardMS24}).
Dallard et al.~\cite{dallard2022secondpaper} conjecture the converse, namely, that every $(\tw,\omega)$-bounded graph class has bounded tree-independence number.
This conjecture was recently disproved by Chudnovsky and Trotignon~\cite{CT24}.
A related research direction in structural graph theory is the study of induced obstructions to bounded tree-independence number; see, for example, the recent works~\cite{dallard2024treewidthversuscliquenumber,abrishami2023tree,chudnovsky2024treeindependencenumberiv,chudnovsky2024treeindependencenumberiii,chudnovsky2024treeindependencenumberii}.

In our opinion, the most interesting application of  \Cref{the:main_alg} lies in the area of graph algorithms for \NP-hard optimization problems. Dallard et al.~\cite{DallardMS24} and Yolov in \cite{Yolov18} have shown that certain $\NP$-hard optimization problems like \textsc{Maximum Independent Set}, \textsc{Graph Homomorphism}, or \textsc{Maximum Induced Matching} problems can be solved in time $n^{\OO(k)}$ if the input graph is given with a tree decomposition of independence number at most~$k$.
Lima et al.~\cite{LMMORS24} extended this idea to generic packing problems in which any two of the chosen subgraphs have to be at pairwise distance at least $d$, for even~$d$.
They also obtained an algorithmic metatheorem for the problem of finding a maximum-weight sparse (bounded chromatic number) induced subgraph satisfying an arbitrary but fixed property expressible in counting monadic second-order logic (\cmsotwo).
In particular, the metatheorem implies polynomial-time solvability of several classical problems like finding the largest induced forest (which is equivalent to \textsc{Minimum Feedback Vertex Set}), finding the largest induced bipartite subgraph (which is equivalent to \textsc{Minimum Odd Cycle Transversal}), finding the maximum number of pairwise disjoint and non-adjacent cycles (\textsc{Maximum Induced Cycle Packing}), and finding the largest induced planar subgraph (which is equivalent to \textsc{Planarization}).

However, the weak spot in all these algorithmic approaches is the requirement that a tree decomposition with bounded independence number is given with the input.
\Cref{the:main_alg} fills this gap by constructing a decomposition of bounded independence number
in time that asymptotically matches or improves the time required to solve all these optimization problems.

\Cref{the:main_alg} appears to be a handy tool in the subarea of computational geometry concerning optimization problems on geometric graphs. Treewidth plays a fundamental role in the design of exact and approximation algorithms on planar graphs (and more generally, $H$-minor-free graphs) \cite{DemaineFHT05jacm,Baker94,Grohe:2003kt}.
The main property of such graphs is that they enjoy the bounded local treewidth property. In other words, any planar graph of a small diameter has a small treewidth. 
A natural research direction is to extend such methods to intersection graphs of geometric objects \cite{FominLS18,LokshtanovPSXZ22}. 
However, even for very ``simple'' objects like unit disks, the corresponding intersection graphs do not have locally bounded treewidth. 
On the other hand, such graphs of bounded diameter have bounded tree-independence number. 
Thus, in many scenarios, the treewidth-based methods on such graphs could be replaced by tree decompositions of bounded independence number.
{In particular,} Galby et al.~\cite{Galby23,DBLP:journals/corr/abs-2402-18352} use \Cref{the:main_alg} for obtaining polynomial-time approximation schemes for several packing and induced subgraph problems on geometric graphs.
It is interesting to note that algorithms on geometric graphs often require geometric representation of a graph. Sometimes, like for unit disk graphs, finding such a representation is a challenging computational task \cite{HlinenyK01}.
In contrast, \Cref{the:main_alg} does not need the geometric properties of objects or their geometric representations and thus could be used for developing so-called robust algorithms~\cite{RaghavanS03} on geometric graphs~\cite{ClarkCJ90}.

In parameterized algorithms, Fomin and Golovach~\cite{FominG21} and
Jacob et al.~\cite{DBLP:journals/algorithmica/JacobPRS22} used tree decompositions where each bag is obtained from a clique by deleting at most $k$ edges or adding at most $k$ vertices, respectively.
These type of decompositions are special types of tree decompositions with bounded independence numbers.

The rest of this paper is organized as follows.
In \Cref{sec:overview}, we overview the proofs of our main results, \Cref{the:main_alg} and \Cref{the:main_nphard}.
In \Cref{sec:preliminaries} we recall definitions and preliminary results.
\Cref{sec:algo} is devoted to the proof of \Cref{the:main_alg} and \Cref{sec:lower-bound} to the proof of \Cref{the:main_nphard}.
In \Cref{sec:hardsep} we show the \NP-hardness of finding separators with bounded independence number.
We conclude in \Cref{sec:concl} with final remarks and open problems.

\section{Overview\label{sec:overview}}
In this section we sketch the proofs of our two main results, \cref{the:main_alg} and \cref{the:main_nphard}, and compare the algorithm of \cref{the:main_alg} to the algorithm of Yolov~\cite{Yolov18}.

\subsection{Outline of the algorithm\label{sec:alg_outline}}
For simplicity, we sketch here a version of \cref{the:main_alg} with approximation ratio 11 instead of 8.
The difference between the 11-approximation and 8-approximation is that for 11-approximation it is sufficient to use 2-way separators, while for 8-approximation we use 3-way separators.

On a high level, our algorithm follows the classical technique of constructing a tree decomposition by repeatedly finding balanced separators.
This technique was introduced by Robertson and Seymour in Graph Minors~XIII~\cite{RobertsonS-GMXIII}, was used for example in~\cite{DBLP:journals/siamcomp/BodlaenderDDFLP16,DBLP:journals/siamcomp/BergBKMZ20,DBLP:journals/algorithmica/JacobPRS22}, and an exposition of it was given in~\cite[Section~7.6.2]{cygan2015parameterized} and in~\cite{DBLP:conf/birthday/Pilipczuk20a}.

The challenge in applying this technique is the need to compute separators that are both balanced and small with respect to the independence numbers of the involved vertex sets. Our main technical contribution is an approximation algorithm for finding such separators.
In what follows, we will sketch an algorithm that, given a graph $G$, a parameter $k$, and a set of vertices $W \subseteq V(G)$ with independence number $\alpha(W) = 9k$, either finds a partition $(S, C_1, C_2)$ of $V(G)$ such that each $S$, $C_1$, and $C_2$ are nonempty, $S$ separates $C_1$ from $C_2$, $\alpha(S) \le 2k$, and $\alpha(W \cap C_i) \le 7k$ for both $i =1,2$, or determines that the tree-independence number of $G$ is larger than~$k$. (For $S\subseteq V(G)$ we use $\alpha(S)$  to denote the independence number of the induced subgraph $G[S]$.)
Our algorithm for finding such balanced separators works in two parts.
First, we reduce finding balanced separators to finding separators, and then we give a 2-approximation algorithm for finding separators.

We emphasize that the balancedness condition is expressed in terms of the independence number of the separated subgraphs and not in terms of their number of vertices.
Hence, at first glance, it is not even clear that small tree-independence number guarantees the existence of such balanced separators.
To prove the existence of balanced separators and to reduce the finding of balanced separators to finding separators between specified sets of vertices, instead of directly enforcing that both sides of the separation have a small independence number, we enforce that both sides of the separation have sufficiently large independence number.
More precisely, we pick an arbitrary independent set $I \subseteq W$ of size $|I| = 9k$. By making use of the properties of tree decompositions,  it is possible to show that there exists a separation $(S, C_1, C_2)$ with $|I \cap C_i| \le 6k$ for $i \in \{1,2\}$ and $\alpha(S) \le k$. Hence  $|I \cap C_i| \ge 2k$ for $i \in \{1,2\}$.
By enforcing the condition  $|I \cap C_i| \ge 2k$ for both $i \in \{1,2\}$, we will have  that $\alpha(W \cap C_i) \le 7k$ for both $i \in \{1,2\}$. (Note that  $\alpha(W \cap C_1) + \alpha(W \cap C_2) \le \alpha(W)$.)
Therefore, to find a balanced separator for $W$ or to conclude that the tree-independence number of $G$ is larger than $k$, it is sufficient to select an arbitrary independent set $I \subseteq W$ with $|I| = 9k$, do $2^{\OO(k)}$ guesses for sets $I \cap C_1$ and $I \cap C_2$, and for each of the guesses search for a separator between the sets $I \cap C_1$ and $I \cap C_2$.

In the separator finding algorithm the input includes two sets $V_1 = I \cap C_1$ and $V_2 = I \cap C_2$, and the task is to find a set of vertices $S$ disjoint from both $V_1$ and $V_2$ separating $V_1$ from $V_2$ with $\alpha(S) \le 2k$ or to conclude that no such separator $S$ with $\alpha(S) \le k$ exists.
Our algorithm works by first using multiple stages of different branching steps, and then arriving at a special case which is solved by rounding a linear program.
We explain some details in what follows.

First, by using iterative compression around the whole algorithm, we can assume that we have a tree decomposition with independence number $\OO(k)$ available.
We show that any set $S \subseteq V(G)$ with $\alpha(S) \le k$ can be covered by $\OO(k)$ bags of the tree decomposition.
This implies that by first guessing the covering bags, we reduce the problem to the case where we search for a separator $S \subseteq R$ for some set $R \subseteq V(G)$ with independence number $\alpha(R) = \OO(k^2)$.

Then, we use a branching procedure to reduce the problem to the case where $R \subseteq N[V_1 \cup V_2]$. In the branching, we select a vertex $v \in R \setminus N[V_1 \cup V_2]$, and branch into three subproblems, corresponding to including $v$ into $V_1$, into $V_2$, or into a partially constructed solution~$S$.
The key observation here is that if we branch on vertices $v \in R \setminus N[V_1 \cup V_2]$, then the branches where $v$ is included in $V_1$ or in $V_2$ reduce the value $\alpha(R \setminus N[V_1 \cup V_2])$.
By first handling the case with $\alpha(R \setminus N[V_1 \cup V_2]) \ge 2k$ by branching on $2k$ vertices at the same time and then branching on single vertices, this branching results in $2^{\OO(\alpha(R))} n^{\OO(k)} = 2^{\OO(k^2)} n^{\OO(k)}$ instances where $R \subseteq N[V_1 \cup V_2]$.

Finally, when we arrive at the subproblem where $R \subseteq N[V_1 \cup V_2]$, we design a \hbox{$2$-approximation} algorithm by rounding a linear program.
For $v\in R$, let us have variables $x_v$ with $x_v = 1$ indicating that $v \in S$ and $x_v = 0$ indicating that $v \notin S$.
Because $R \subseteq N[V_1 \cup V_2]$, the fact that $S \subseteq R$ separates $V_1$ from $V_2$ can be encoded by only using inequalities of form $x_v + x_u \ge 1$.
To bound the independence number of $S$, for every independent set $I\subseteq R$ of size $|I| = 2k+1$ we add a constraint that $\sum_{v \in I} x_v \le k$.
Now, a separator $S$ with $\alpha(S) \le k$ corresponds to an integer solution.
We then find a fractional solution and round $x_v$ to $1$ if $x_v \ge 1/2$ and to $0$ otherwise.
Note that this satisfies the $x_v + x_u \ge 1$ constraints, so the rounded solution corresponds to a separator.
To bound the independence number of the rounded solution, note that $\sum_{v \in I} x_v \le 2k$  for independent sets $I$ of size $|I| = 2k+1$, therefore implying that $\alpha(S) \le 2k$.

\paragraph*{Comparison to the algorithm of Yolov.}
Both our algorithm and the algorithm of Yolov~\cite{Yolov18} are based on the general approach of constructing a tree decomposition by repeatedly finding balanced separators, originating from Graph Minors~XIII~\cite{RobertsonS-GMXIII}.
The main difference between our algorithm and the algorithm of Yolov is the approach for finding separators with bounded independence number.
The separator algorithm of Yolov has running time $n^{\OO(k)}$ but an approximation ratio of $\OO(k^2)$, which in the end manifests in both the overall running time of $n^{\OO(k^3)}$ and approximation ratio of $\OO(k^2)$.
Our algorithm could be thought of as improving the separator finding algorithm to have an approximation ratio of $2$, resulting in improving both the overall running time to $2^{\OO(k^2)} n^{\OO(k)}$ and the approximation ratio to~$8$.
Our separator finding algorithm appears to use completely different techniques than the separator finding algorithm of Yolov.
Our approach crucially uses the fact that we have an approximately optimal tree decomposition available by iterative compression, while the approach of Yolov does not use a tree decomposition, but instead employs a greedy edge addition lemma combined with structural results on separators with bounded independence number, and finally uses 2-SAT to find the separator.

\subsection{Outline of the \texorpdfstring{\NP}{NP}-hardness proof}
We explain the general idea of the reduction behind \Cref{the:main_nphard} in a somewhat inversed order, in particular by making a sequence of observations about a particular special case of the problem of determining if $\tin(G) \le 4$, culminating in observing that it corresponds to \textsc{3-coloring}.

Let $G$ be a graph that does not contain cliques of size four.
We create a graph $G'$ by first taking the disjoint union of two copies $G_1$ and $G_2$ of the complement $\overline{G}$, and then adding a matching between $V(G_1)$ and $V(G_2)$ by connecting the corresponding copies of the vertices.
Then, we enforce that the sets $V(G_1)$ and $V(G_2)$ are bags in any tree decomposition of $G'$ of the independence number at most $4$ by adding $5$ vertices $v_1, \ldots, v_5$ with $N(v_i) = V(G_1)$ and 5 vertices $u_1, \ldots, u_5$ with $N(u_i) = V(G_2)$ (using \cref{obs:folk}).
Now, the problem of determining whether $\tin(G') \le 4$ boils down to determining if we can construct a path decomposition with independence number at most $4$ with a bag $V(G_1)$ on one end and a bag $V(G_2)$ on the other end (the additional vertices $v_1, \ldots, v_5$ and $u_1, \ldots, u_5$ can be immediately ignored after the bags $V(G_1)$ and $V(G_2)$ have been enforced).

The problem of constructing such a path decomposition corresponds to the problem of finding an ordering according to which the vertices in $V(G_1)$ should be eliminated when traversing the decomposition from the bag of $V(G_1)$ to the bag of $V(G_2)$.
In particular, because the edges between $V(G_1)$ and $V(G_2)$ are a matching, we can assume that in an optimal path decomposition a vertex of $V(G_2)$ is introduced right before the corresponding vertex in $V(G_1)$ is eliminated.
From the viewpoint of the original graph $G$, this problem now corresponds to finding an ordering $v_1, \ldots, v_n$ of $V(G)$ such that for every $i \in [n-1]$ it holds that $\omega(G[\{v_1, \ldots, v_i\}]) + \omega(G[\{v_{i+1}, \ldots, v_n\}]) \le 4$, where $\omega$ denotes the maximum size of a clique in a graph.
This problem in turn corresponds to determining if $G$ has two disjoint subsets $U_1,U_2 \subseteq V(G)$ such that for both $i = 1,2$ it holds that the vertex cover number of $G[U_i]$ is at most 1, and $U_i$ intersects every triangle of~$G$.
We show a reduction from \textsc{3-coloring} to this problem, which finishes the \NP-hardness proof.

\section{Preliminaries}\label{sec:preliminaries}

We denote the vertex set and the edge set of a graph $G = (V,E)$ by $V(G)$ and $E(G)$, respectively.
The \emph{neighborhood} of a vertex $v$ in $G$ is the set $N_G(v)$ of vertices adjacent to $v$ in $G$, and the \emph{closed neighborhood} of $v$ is the set $N_G[v] = N_G(v) \cup \{v\}$.
These concepts are extended to sets $X\subseteq V(G)$ so that $N_G[X]$ is defined as the union of all closed neighborhoods of vertices in $X$, and $N_G(X)$ is defined as the set $N_G[X]\setminus X$.
The \emph{degree} of $v$, denoted by $d_G(v)$, is the cardinality of the set $N_G(v)$.
When there is no ambiguity, we may omit the subscript $G$ in the notations of the degree, and open and closed neighborhoods, and thus simply write $d(v)$, $N(v)$, and $N[v]$, respectively.

Given a set $X \subseteq V(G)$, we denote by $G[X]$ the subgraph of $G$ induced by~$X$.
We also write $G \setminus X = G[V(G) \setminus X]$.
Similarly, given a vertex $v \in V(G)$, we denote by $G \setminus v$ the graph obtained from $G$ by deleting~$v$.
The \emph{complement} of a graph $G$ is the graph $\overline{G}$ with vertex set $V(G)$ in which two distinct vertices are adjacent if and only if they are non-adjacent in~$G$.
For a positive integer $n$, we denote the $n$-vertex complete graph, path, and cycle by $K_n$, $P_n$, and $C_n$, respectively.
For positive integers $m$ and $n$ we denote by $K_{m,n}$ the complete bipartite graph with parts of sizes $m$ and~$n$.
Given two graphs $G$ and $H$, we say that $G$ is \emph{$H$-free} if no induced subgraph of $G$ is isomorphic to $H$.
A graph $G$ is \emph{chordal} if it has no induced cycles of length at least four.

A \emph{clique} (resp.\ \emph{independent set}) in a graph $G$ is a set of pairwise adjacent (resp.\ non-adjacent) vertices.
The \emph{independence number} of $G$, denoted by $\alpha(G)$, is the maximum size of an independent set in~$G$.
For a set of vertices $X \subseteq V(G)$, the independence number of $X$ is $\alpha(X) = \alpha(G[X])$.

Let $(V_1, V_2, \ldots, V_t)$ be a tuple of disjoint subsets of $V(G)$.
A \emph{$(V_1, V_2, \ldots, V_t)$-separator} is a set $S \subseteq V(G)$ such that $S \cap V_i = \emptyset$ for each $i \in [t]$, and in the graph $G \setminus S$ there is no path from $V_i$ to $V_j$ for all pairs $i \neq j$.
Such a separator is sometimes called a $t$-way separator.
Note that if there is an edge with an endpoint in $V_i$ and the other in $V_j$, for some $i \neq j$, then no $(V_1, V_2, \ldots, V_t)$-separator exists.

A \emph{tree decomposition} of a graph $G$ is a pair $\td = (T, \{X_t\}_{t\in V(T)})$ where $T$ is a tree and every node $t$ of $T$ is assigned a vertex subset $X_t\subseteq V(G)$ called a bag such that the following conditions are satisfied:
(1) every vertex of $G$ is in at least one bag, (2) for every edge $uv\in E(G)$ there exists a node $t\in V(T)$ such that $X_t$ contains both $u$ and $v$, and (3) for every vertex $u\in V(G)$ the subgraph $T_u$ of $T$ induced by the set $\{t\in V(T): u\in X_t\}$ is connected (that is, a tree).
The \emph{independence number} of a tree decomposition $\td = (T, \{X_t\}_{t\in V(T)})$ of a graph $G$, denoted by $\alpha(\td)$, is defined as follows:
\[\alpha(\td) = \max_{t\in V(T)} \alpha(X_t)\,.\]
The \emph{tree-independence number} of $G$, denoted by $\tin(G)$, is the minimum independence number among all tree decompositions of~$G$.

\section{An 8-approximation algorithm for tree-independence number\label{sec:algo}}
In this section we prove \cref{the:main_alg}, that is, we give a $2^{\OO(k^2)}n^{\OO(k)}$-time algorithm for either computing tree decompositions with independence number at most $8k$ or deciding that the tree-independence number of the graph is more than~$k$.
Our algorithm consists of three parts.
First, we give a $2^{\OO(k^2)}n^{\OO(k)}$-time $2$-approximation algorithm for finding 3-way separators with independence number at most $k$, with the assumption that a tree decomposition with the  independence number $\OO(k)$ is given with the input.
Then, we apply this separator finding algorithm to find balanced separators, and then apply balanced separators in the fashion of the Robertson-Seymour treewidth approximation algorithm~\cite{RobertsonS-GMXIII} to construct a tree decomposition with independence number at most~$8k$.
The requirement for having a tree decomposition with independence number $\OO(k)$ as an input in the separator algorithm is satisfied by iterative compression (see, e.g.,~\cite{cygan2015parameterized}), as we explain at the end of~\Cref{subsec:constructing}.

The presentation of the algorithm in this section is in reverse order compared to the presentation we gave in  \Cref{sec:alg_outline}.

\subsection{Finding approximate separators\label{subsec:apxsep}}
In this subsection, we show the following theorem.

\begin{theorem}
\label{the:sepfinding}
There is an algorithm that, given a graph $G$, an integer $k$, a tree decomposition $\td$ of $G$ with independence number $\alpha(\td) = \OO(k)$, and three disjoint sets of vertices $V_1, V_2, V_3 \subseteq V(G)$, in time $2^{\OO(k^2)} n^{\OO(k)}$ either reports that no $(V_1, V_2, V_3)$-separator with independence number at most $k$ exists, or returns a $(V_1, V_2, V_3)$-separator with independence number at most~$2k$.
\end{theorem}

To prove \Cref{the:sepfinding}, we define a more general problem that we call \textsc{partial $3$-way $\alpha$-separator}, which is the same as the problem of \cref{the:sepfinding} except that two sets $S_0$ and $R$ are given with the input, and we are only looking for separators $S$ with $S_0 \subseteq S \subseteq S_0 \cup R$.

\begin{definition}[\textsc{Partial 3-way $\alpha$-separator}]
An instance of \textsc{partial 3-way $\alpha$-separator} is a 5-tuple $(G, (V_1, V_2, V_3), S_0, R, k)$, where $G$ is a graph, $V_1, V_2, V_3, S_0$, and $R$ are disjoint subsets of $V(G)$, and $k$ is an integer.
A solution to \textsc{partial 3-way $\alpha$-separator} is a $(V_1, V_2, V_3)$-separator $S$ such that $S_0 \subseteq S \subseteq S_0 \cup R$.
A $2$-approximation algorithm for \textsc{partial 3-way $\alpha$-separator} either returns a solution $S$ with $\alpha(S) \le 2k$ or determines that there is no solution $S$ with $\alpha(S) \le k$.
\end{definition}

We give a 2-approximation algorithm for \textsc{partial 3-way $\alpha$-separator}.
Then \cref{the:sepfinding} will follow by setting $S_0 = \emptyset$ and $R = V(G) \setminus (V_1 \cup V_2 \cup V_3)$.

We give our 2-approximation algorithm by giving 2-approximation algorithms for increasingly more general cases of \textsc{partial 3-way $\alpha$-separator}.
First, we give a linear programming-based 2-approximation algorithm for the special case when $R \subseteq N(V_1 \cup V_2 \cup V_3)$.
Then, we use branching and the first algorithm to give a 2-approximation algorithm for the case when $\alpha(R) = \OO(k^2)$.
Then, we use the input tree decomposition to reduce the general case to the case where $\alpha(R) = \OO(k^2)$.

Let us make some observations about trivial instances of \textsc{partial 3-way $\alpha$-separator}.
First, we can assume that $\alpha(S_0) \le k$, as otherwise any solution $S$ has $\alpha(S) > k$ and we can immediately return ``no''.
All our algorithms include an $n^{\OO(k)}$ factor in the time complexity, so it can be assumed that this condition is always tested.
Then, we can also always return ``no'' if some vertex of $V_i$ is adjacent to a vertex of $V_j$, where $i \neq j$.
This can be checked in polynomial time, so we can assume that this condition is always tested.
Note that testing this condition implies that $N(V_1 \cup V_2 \cup V_3) = N(V_1) \cup N(V_2) \cup N(V_3)$. 
For simplicity, we write $N(V_1 \cup V_2 \cup V_3)$ as it is shorter.

\medskip

We start with the linear programming-based $2$-approximation algorithm for the case when $R \subseteq N(V_1 \cup V_2 \cup V_3)$.

\begin{lemma}
\label{lem:sepapxLP}
There is an $n^{\OO(k)}$-time $2$-approximation algorithm for \textsc{partial $3$-way $\alpha$-separator} when $R \subseteq N(V_1 \cup V_2 \cup V_3)$.
\end{lemma}

\begin{proof}
First, note that we may assume that for all $i,j\in \{1,2,3\}$, $i\neq j$, there is no path in the graph $G\setminus (R\cup S_0)$ between a vertex of $V_i$ and a vertex of $V_j$, since otherwise the given instance has no solution at all.
Furthermore, under this assumption, we can safely replace each set $V_i$ with the connected component of $G \setminus (R \cup S_0)$ containing $V_i$ in the graph $G\setminus (R\cup S_0)$.
We thus arrive at an instance such that $R\subseteq N(V_1 \cup V_2 \cup V_3) \subseteq R \cup S_0$.
We then make an integer programming formulation of the problem (with $n^{\OO(k)}$ constraints) and show that it gives a 2-approximation by rounding a fractional solution.

For every vertex $v \in R \cup S_0$ we introduce a variable $x_v$, with the interpretation that $x_v = 1$ if $v \in S$ and $x_v = 0$ otherwise.
We force $S_0$ to be in the solution by adding constraints $x_v = 1$ for all $v \in S_0$.
Then, we say that a pair $(v_i, v_j)$ of vertices with $v_i \in N(V_i)$, $v_j \in N(V_j)$, $i \neq j$, is \emph{connected} if there is a $v_i,v_j$-path in the graph $G \setminus ((S_0 \cup R) \setminus \{v_i, v_j\})$.
For any pair $(v_i, v_j)$ of connected vertices we introduce a constraint $x_{v_i} + x_{v_j} \ge 1$ indicating that $v_i$ or $v_j$ should be selected to the solution.
Note that here it can happen that $v_i = v_j$, corresponding to the case when $v_i \in N(V_i) \cap N(V_j)$, resulting in a constraint forcing $v_i$ to be in the solution.
Finally, for every independent set $I \subseteq R \cup S_0$ of size $|I| = 2k+1$, we introduce a constraint $\sum_{v \in I} x_v \le k$.

We observe that any solution $S$ with $\alpha(S) \le k$ corresponds to a solution to the integer program.
In particular, any $(V_1,V_2,V_3)$-separator $S$ must satisfy the $x_{v_i} + x_{v_j} \ge 1$ constraints for connected pairs $(v_i,v_j)$, as otherwise there would be a $V_i$,$V_j$-path in $G \setminus S$, and a separator with $\alpha(S) \le k$ satisfies the independent set constraints as otherwise it would contain an independent set larger than~$k$.

Then, we show that any integral solution that satisfies the $x_{v_i} + x_{v_j} \ge 1$ constraints for connected pairs $(v_i, v_j)$ and the constraints $x_v = 1$ for $v \in S_0$ corresponds to a $(V_1, V_2, V_3)$-separator.
Suppose that all such constraints are satisfied by an integral solution corresponding to a set $S$, but there is a path from $V_i$ to $V_j$ with $i \neq j$ in $G \setminus S$.
Take a shortest path connecting $N(V_i) \setminus S$ to $N(V_j) \setminus S$ in $G \setminus S$ for any $i \neq j$, and note that the intermediate vertices of such a path do not intersect $N(V_1 \cup V_2 \cup V_3)$, as otherwise we could get a shorter path, and therefore do not intersect~$R$.
By $S_0 \subseteq S$, we have that the intermediate vertices do not intersect $R \cup S_0$, so this path is in fact a path in the graph $G \setminus ((S_0 \cup R) \setminus \{v_i, v_j\})$, where $v_i \in N(V_i) \setminus S$ is the first vertex and $v_j \in N(V_j) \setminus S$ is the last vertex.
However, in this case $(v_i, v_j)$ is a connected pair and we have a constraint $x_{v_i} + x_{v_j} \ge 1$, which would be violated by such an integral solution, so we get a contradiction.

We solve the linear program in polynomial time (which is $n^{\OO(k)}$ as the number of variables is $\OO(n)$ and the number of constraints is $n^{\OO(k)}$), and round the solution by rounding $x_v$ to 1 if $x_v \ge 1/2$ and otherwise to 0.
This rounding will satisfy the $x_{v_i} + x_{v_j} \ge 1$ constraints for connected pairs, so the resulting solution corresponds to a separator.
Then, by the constraints $\sum_{v \in I} x_v \le k$ for independent sets $I$ of size $2k+1$, the rounded solution will satisfy $\sum_{v \in I} x_v \le 2k$ for independent sets $I$ of size $2k+1$.
Therefore, there are no independent sets of size $2k+1$ in the resulting solution, so its independence number is at most~$2k$.
\end{proof}

We will pipeline \cref{lem:sepapxLP} with branching to obtain a 2-approximation algorithm for \textsc{partial 3-way $\alpha$-separator} in a setting where $\alpha(R)$ is small.
In our final algorithm, $\alpha(R)$ will be bounded by $\OO(k^2)$, and this is the part that causes the $2^{\OO(k^2)}$ factor in the time complexity.

Let us observe that we can naturally branch on vertices in $R$ in instances of \textsc{partial 3-way $\alpha$-separator}.

\begin{lemma}[Branching]
\label{obs:branch}
Let $\mathcal{I} = (G, (V_1, V_2, V_3), S_0, R, k)$ be an instance of \textsc{partial 3-way $\alpha$-separator}, and let $v \in R$.
If $S$ is a solution of $\mathcal{I}$, then $S$ is also a solution of at least one of
\begin{enumerate}
\item $(G, (V_1 \cup \{v\}, V_2, V_3), S_0, R \setminus \{v\}, k)$,
\item $(G, (V_1, V_2 \cup \{v\}, V_3), S_0, R \setminus \{v\}, k)$,
\item $(G, (V_1, V_2, V_3 \cup \{v\}), S_0, R \setminus \{v\}, k)$, or
\item $(G, (V_1, V_2, V_3), S_0 \cup \{v\}, R \setminus \{v\}, k)$.
\end{enumerate}
Moreover, any solution of any of the instances 1-4 is also a solution of $\mathcal{I}$.
\end{lemma}

\begin{proof}
If $S$ is a solution of $\mathcal{I}$, we can partition $V(G) \setminus S$ into parts $V'_1 \supseteq V_1$, $V'_2 \supseteq V_2$, and $V'_3 \supseteq V_3$ by setting $V_i'$ equal to the union of the vertex sets of connected components of $G \setminus S$ containing vertices of $V_i$
for all $i\in\{1,2,3\}$, and including the vertices in connected components containing no vertices of $V_1, V_2, V_3$ into $V'_1$, resulting in a partition $(V'_1, V'_2, V'_3)$ of $V(G) \setminus S$ such that $S$ is a $(V'_1, V'_2, V'_3)$-separator.
Therefore, the first branch corresponds to when $v \in V'_1$, the second when $v \in V'_2$, the third when $v \in V'_3$, and the fourth when $v \in S$.

The direction that any solution of any of the instances 1-4 is also a solution of the original instance is immediate from the fact that we do not remove vertices from any $V_i$ or~$S_0$.
\end{proof}

\cref{obs:branch} allows branching into four subproblems, corresponding to the situation whether a vertex from  $R$ goes to $V_1$, $V_2$, $V_3$, or to~$S_0$.
Now, our goal is to branch until we derive an instance with $R \subseteq N(V_1 \cup V_2 \cup V_3)$, which can be solved by \cref{lem:sepapxLP}.
In particular, we would like to branch on vertices $v \in R \setminus N(V_1 \cup V_2 \cup V_3)$.
The following lemma shows that we can use $\alpha(R \setminus N(V_1 \cup V_2 \cup V_3))$ as a measure of progress in the branching.

\begin{lemma}
\label{obs:branchadec}
For any vertex $v \in R \setminus N(V_1 \cup V_2 \cup V_3)$ it holds that $\alpha(R \setminus (N[v] \cup N(V_1 \cup V_2 \cup V_3))) < \alpha(R \setminus N(V_1 \cup V_2 \cup V_3))$.
\end{lemma}
\begin{proof}
If $\alpha(R \setminus (N[v] \cup N(V_1 \cup V_2 \cup V_3))) \ge \alpha(R \setminus N(V_1 \cup V_2 \cup V_3))$, then we could take an independent set $I \subseteq R \setminus (N[v] \cup N(V_1 \cup V_2 \cup V_3))$ such that
$|I| = \alpha(R \setminus N(V_1 \cup V_2 \cup V_3))$, and construct an independent set $I \cup \{v\} \subseteq R \setminus N(V_1 \cup V_2 \cup V_3)$ of size $|I \cup \{v\}| > \alpha(R \setminus N(V_1 \cup V_2 \cup V_3))$, which is a contradiction.
\end{proof}

\cref{obs:branchadec} implies that when branching on a vertex $v \in R \setminus N(V_1 \cup V_2 \cup V_3)$, the branches where $v$ is moved to $V_1$, $V_2$, or $V_3$ decrease $\alpha(R \setminus N(V_1 \cup V_2 \cup V_3))$.
This will be the main idea of our algorithm for the case when $\alpha(R)$ is bounded, which we give next.

\begin{lemma}
\label{lem:sepfind_boundalpha}
There is a $2^{\OO(\alpha(R))} n^{\OO(k)}$-time $2$-approximation algorithm for \textsc{partial 3-way $\alpha$-separator}.
\end{lemma}
\begin{proof}
We give a branching algorithm, whose base case is the case when $R \subseteq N(V_1 \cup V_2 \cup V_3)$, which is solved by \cref{lem:sepapxLP}.
The main idea is to analyze the branching by the paramter $\alpha(R \setminus N(V_1 \cup V_2 \cup V_3))$, in particular with $\alpha(R \setminus N(V_1 \cup V_2 \cup V_3)) = 0$ corresponding to the base case.
The branching itself will be ``exact'' in the sense that all four cases of \cref{obs:branch} are always included, in particular, the 2-approximation is caused only by the application of \cref{lem:sepapxLP}.

First, while $\alpha(R \setminus N(V_1 \cup V_2 \cup V_3)) \ge 2k$, which can be checked in $n^{\OO(k)}$ time, we branch as follows.
We select an independent set $I \subseteq R \setminus N(V_1 \cup V_2 \cup V_3)$ of size $|I| = 2k$, and for all of its vertices we branch on whether to move it into  $V_1$, $V_2$, $V_3$, or $S_0$, i.e., according to \cref{obs:branch}.
Because $I$ is an independent set, at most $k$ of the vertices can go to $S_0$, so at least $k$ go to $V_1$, $V_2$, or~$V_3$.
Also for the reason that $I$ is an independent set, \cref{obs:branchadec} can be successively applied for all of the vertices that go to $V_1$, $V_2$, or~$V_3$.
Therefore, this decreases $\alpha(R \setminus N(V_1 \cup V_2 \cup V_3))$ by at least $k$, so the depth of this recursion is at most $\alpha(R)/k$, so the total number of nodes in this branching tree is at most $(4^{2k})^{\alpha(R)/k} = 2^{\OO(\alpha(R))}$.

Then, we can assume that $\alpha(R \setminus N(V_1 \cup V_2 \cup V_3)) < 2k$.
We continue with a similar branching, this time branching on single vertices.
In particular, as long as $R \setminus N(V_1 \cup V_2 \cup V_3)$ is nonempty, we select a vertex in it and branch on whether to move it into $V_1$, $V_2$, $V_3$, or~$S_0$.
To analyze the size of this branching tree, note that by \cref{obs:branchadec}, when moving a vertex into $V_1$, $V_2$, or $V_3$, the value of $\alpha(R \setminus N(V_1 \cup V_2 \cup V_3))$ decreases by at least one.
Therefore, as initially $\alpha(R \setminus N(V_1 \cup V_2 \cup V_3)) < 2k$, any root-to-leaf path of the branching tree contains less than $2k$ edges that correspond to such branches, and therefore any root-leaf path can be characterized by specifying the $<2k$ indices corresponding to such edges, and whether these indices correspond to $V_1$, $V_2$, or~$V_3$.
The length of any root-leaf path is at most $n$ because $|R \setminus N(V_1 \cup V_2 \cup V_3)|$ decreases in every branch, and therefore the number of different root-to-leaf paths is at most $n^{2k} 3^{2k} = n^{\OO(k)}$, and therefore the total number of nodes in this branching tree is $n^{\OO(k)}$.

Therefore, the total size of both of the branching trees put together is $2^{\OO(\alpha(R))} n^{\OO(k)}$, so our algorithm works by $2^{\OO(\alpha(R))} n^{\OO(k)}$ applications of \cref{lem:sepapxLP}, resulting in a $2^{\OO(\alpha(R))} n^{\OO(k)}$-time algorithm.
\end{proof}

Finally, what is left is to reduce the general case to the case where $\alpha(R) = \OO(k^2)$.
We do not know if this can be done in general, but we manage to do it by using a tree decomposition with independence number $\OO(k)$.
To this end, we show the following lemma.
Given a graph $G$, a tree decomposition $\td$ of $G$, and a set $W\subseteq V(G)$, we say that $W$ is \emph{covered} by a set $\mathcal{B}$ of bags of $\td$ if every vertex in $W$ is contained in at least one of the bags in $\mathcal{B}$.
The lemma generalizes the well-known fact that any clique $W$ of $G$, that is, a set with $\alpha(W)=1$, is covered by a single bag of the tree decomposition.

\begin{lemma}
\label{lem:bag_cover_is}
Let $G$ be a graph, $\td$ a tree decomposition of $G$, and $W$ a nonempty set of vertices of~$G$.
Then $W$ is covered by a set of at most  $2 \alpha(W)-1$ bags of~$\td$.
\end{lemma}

\begin{proof}
We denote $\td = (T, \{X_t\}_{t \in V(T)})$.
Every edge $ab \in E(T)$ corresponds to a partition $(X_a \cap X_b, C_a, C_b)$ of $V(G)$, where $C_a$ is the set of vertices of $G \setminus (X_a \cap X_b)$ in the bags of $\td$ that are closer to $a$ than $b$, $C_b$ is the set of vertices of $G \setminus (X_a \cap X_b)$ in the bags of $\td$ that are closer to $b$ than $a$, and $X_a \cap X_b$ is a $(C_a, C_b)$-separator.

First, assume that for every edge $ab$ either $C_a \cap W = \emptyset$ or $C_b \cap W = \emptyset$.
If both $C_a \cap W =C_b \cap W = \emptyset$, then $W\subseteq X_a\cap X_b$, and we cover $W$ by a single bag $X_a$ (or $X_b$).
Thus, we may assume that for every edge exactly one of the sets  $C_a \cap W$ and $C_b \cap W$  is nonempty.
We orient the edge $ab$ from $a$ towards $b$ if $C_b \cap W \neq \emptyset$, and from $b$ to $a$ if $C_a \cap W \neq \emptyset$.
Because $T$ is a tree, there exists a node $t \in V(T)$ such that all edges incident with $t$ are oriented towards~$t$.
This implies that $X_t$ covers $W$ because otherwise some edge would be oriented away from~$t$.

Note that if $\alpha(W) = 1$, then $W$ is a clique and indeed for every edge $ab$  either $C_a \cap W = \emptyset$ or $C_b \cap W = \emptyset$.
The remaining case is that $\alpha(W) \ge 2$ and there exists an edge $ab$ such that both $|C_a \cap W| \ge 1$ and $|C_b \cap W| \ge 1$ hold.
In this case, we use induction on $\alpha(W)$.
Since $X_a \cap X_b$ is a $(C_a, C_b)$-separator, we have that $\alpha(C_a \cap W) + \alpha(C_b \cap W) \le \alpha(W)$.
In particular, note that $\alpha(C_a \cap W) \leq \alpha(W) - \alpha(C_b \cap W) < \alpha(W)$, and similarly $\alpha(C_b \cap W) < \alpha(W)$.
Therefore, we can take the union of a smallest set of bags covering $C_a \cap W$, a smallest set of bags covering $C_b \cap W$, and the bag~$X_a$. 
By induction, this set of bags covering $W$ contains at most $2\alpha(W \cap C_a)-1 + 2\alpha(W \cap C_b)-1 + 1 \le 2\alpha(W)-1$ bags.
\end{proof}

With \cref{lem:bag_cover_is}, we can use a tree decomposition $\td$ with independence number $\alpha(\td) = \OO(k)$ to 2-approximate the general case of \textsc{partial 3-way $\alpha$-separator} as follows.
By \cref{lem:bag_cover_is}, any solution $S$ with $\alpha(S) \le k$ is covered by at most $2k-1$ bags of~$\td$.
Therefore, with $\td$ available (and having $n^{\OO(1)}$ bags by standard arguments), we can guess a set of at most $2k-1$ bags of $\td$ that cover $S$, and intersect $R$ by the union of these bags, resulting in $\alpha(R) \le \alpha(\td) (2k-1) = \OO(k^2)$.
Therefore, we solve the general case by $n^{\OO(k)}$ applications of \cref{lem:sepfind_boundalpha} with $\alpha(R) = \OO(k^2)$, resulting in a total time complexity of $2^{\OO(k^2)} n^{\OO(k)}$.
This completes the proof of \cref{the:sepfinding}.

\subsection{From separators to balanced separators}
In this subsection we show that \cref{the:sepfinding} can be used to either find certain balanced separators for sets $W \subseteq V(G)$ or to show that the tree-independence number of the graph is large.

To enforce the ``balance'' condition, we cannot directly enforce that the separator separates a given set $W$ into sets of small independence numbers.  (This would result in time complexity exponential in $|W|$ instead of $\alpha(W)$.)
Instead, we  fix a maximum independent set in $W$, and enforce that this independent set is separated in a balanced manner.
As long as the independent set is large enough, this will enforce that the separator is balanced also with respect to $\alpha$.
The following lemma, which is an adaptation of a well-known lemma given in Graph~Minors~II~\cite{DBLP:journals/jal/RobertsonS86}, is the starting point of this approach.

\begin{lemma}
\label{lem:sep_exist}
Let $G$ be a graph with the tree-independence number at most $k$ and $I \subseteq V(G)$ an independent set.
Then there exists a partition $(S, C_1, C_2, C_3)$ of $V(G)$ (where some $C_i$ can be empty) such that $S$ is a $(C_1, C_2, C_3)$-separator, $\alpha(S) \le k$, and $|I \cap (C_i \cup C_j)| \ge |I|/2-k$ for any pair $i,j \in \{1,2,3\}$ with $i \neq j$.
\end{lemma}
\begin{proof}
Let $\td = (T, \{X_t\}_{t \in V(T)})$ be a tree decomposition of $G$ with $\alpha(\td)\le k$.
As in \Cref{lem:bag_cover_is}, we introduce the following notation.
Every edge $ab \in E(T)$ of $T$ corresponds to a partition $(X_a \cap X_b, C_a, C_b)$ of $V(G)$, where $C_a$ is the set of vertices of $G \setminus (X_a \cap X_b)$ in the bags of $\td$ that are closer to $a$ than $b$, $C_b$ is the set of vertices of $G \setminus (X_a \cap X_b)$ in the bags of $\td$ that are closer to $b$ than $a$, and $X_a \cap X_b$ is a $(C_a, C_b)$-separator.
We orient the edge $ab$ from $a$ to $b$ if $|C_b \cap I| > |I|/2$, from $b$ to $a$ if $|C_a \cap I| > |I|/2$, and otherwise arbitrarily.
Now, because $T$ is a tree, there exists a node $t \in V(T)$ such that all of its incident edges are oriented towards it.
Therefore, for all connected components $C$ of $G \setminus X_t$, we see that $|V(C)\cap I| \le |I|/2$, as otherwise some edge would be oriented out of~$t$.

Let $C_1, C_2, \dots$ be the vertex sets of connected components of $G \setminus X_t$.
As long as the number of such sets is at least 4, we can take two of them with the smallest values of $|C_j \cap I|$ and merge them: the obtained set $C^*$ is such that $|C^* \cap I| \leq |I|/2$.
In the end, we get a partition $(X_t, C_1, C_2, C_3)$ of $V(G)$ such that $X_t$ is a $(C_1,C_2,C_3)$-separator, $\alpha(X_t) \le k$, and $|I \cap C_i| \le |I|/2$ for all $i \in \{1,2,3\}$.
Then, because $|I \cap C_i| \le |I|/2$, we have that  $|I \cap (V(G) \setminus C_i)| \ge |I|/2$.
Therefore, since $|I \cap X_t| \le k$, it follows that $|I \cap (C_j \cup C_\ell)| \ge |I|/2-k$, where $(i,j,\ell)$ is any permutation of the set $\{1,2,3\}$.
\end{proof}

Next we use \cref{lem:sep_exist} together with the separator algorithm of \cref{the:sepfinding} to design an algorithm for finding $\alpha$-balanced separators of sets $W$ with $\alpha(W) = 6k$.

\begin{lemma}
\label{lem:decomp_sepfind}
There is an algorithm that, for a given graph $G$, an integer $k$, a tree decomposition of $G$ with independence number $\OO(k)$, and a set $W \subseteq V(G)$ with $\alpha(W) = 6k$, in time $2^{\OO(k^2)} n^{\OO(k)}$ either concludes that the tree-independence number of $G$ is larger than $k$, or finds a partition $(S, C_1, C_2, C_3)$ of $V(G)$ such that $S$ is a $(C_1,C_2,C_3)$-separator, $\alpha(S) \le 2k$, at most one of $C_1$, $C_2$, and $C_3$ is empty, and $\alpha(W \cap C_i) \le 4k$ for each $i \in \{1,2,3\}$.
\end{lemma}

\begin{proof}
First, we take an arbitrary independent set $I \subseteq W$ of size $|I| = 6k$, which can be found in $n^{\OO(k)}$ time.
If the tree-independence number of $G$ is at most $k$, then, by \cref{lem:sep_exist}, there exists a partition $(S, C_1, C_2, C_3)$ of $V(G)$ such that $S$ is a $(C_1,C_2,C_3)$-separator, $\alpha(S) \le k$, and $|I \cap (C_i \cup C_j)| \ge |I|/2-k \ge 2k$ for any pair $i,j \in \{1,2,3\}$ with $i \neq j$.
We guess the intersection of such a partition with $I$, in particular we guess the partition $(S \cap I, C_1 \cap I, C_2 \cap I, C_3 \cap I)$, immediately enforcing that it satisfies the constraints $|I \cap (C_i \cup C_j)| \ge 2k$ for $i \neq j$.

For each such guess, we use \cref{the:sepfinding} to either find a $(C_1 \cap I, C_2 \cap I, C_3 \cap I)$-separator with independence number at most $2k$ or to decide that no such separator with independence number at most $k$ exists.
The set $S$ of the partition guaranteed by \cref{lem:sep_exist} is indeed a $(C_1 \cap I, C_2 \cap I, C_3 \cap I)$-separator with independence number at most $k$, so if the algorithm reports for every guess that no such separator exists, we return that $G$ has tree-independence number larger than~$k$.

Otherwise, for some guess a $(C_1 \cap I, C_2 \cap I, C_3 \cap I)$-separator $S'$ with $\alpha(S') \le 2k$ is found, and we return the partition $(S', C'_1, C'_2, C'_3)$, where, for $i\in \{1,2\}$, the set $C'_i$ is the union of the vertex sets of components of $G \setminus S'$ that contain a vertex of $C_i \cap I$, and $C_3' = V(G) \setminus (S'\cup C'_1 \cup C'_2)$.
Also, because $|I \cap (C_i' \cup C_j')| \ge 2k$ for any pair $i \neq j$, the set $C_i'$ cannot be empty for more than one~$i$.

The algorithm works by first finding an independent set of size $6k$ and then using the algorithm of \cref{the:sepfinding} at most $4^{6k}$ times, so the total time complexity is $\OO(n^{6k}) + 4^{6k} \cdot 2^{\OO(k^2)} n^{\OO(k)} = 2^{\OO(k^2)} n^{\OO(k)}$.
\end{proof}

\subsection{Constructing the decomposition}\label{subsec:constructing}

Everything is prepared for the final step of the proof---the algorithm constructing a tree decomposition with independence number at most $8k$ by using the balanced separator algorithm of \cref{lem:decomp_sepfind}.
Our algorithm constructs a tree decomposition from the root to the leaves by maintaining an ``interface'' $W$ and breaking it with balanced separators.
This is a common strategy used for various algorithms for constructing tree decompositions and branch decompositions.
In our case, perhaps the largest hurdle in the proof is the analysis that the size of the recursion tree and the constructed decomposition is polynomial in~$n$.

A rooted tree decomposition is a tree decomposition where one node is designated as the root.

\begin{lemma}
\label{lem:maindecompalg}
There is an algorithm that, for a given graph $G$, an integer $k$, a tree decomposition of $G$ with independence number $\OO(k)$, and a set $W \subseteq V(G)$ with $\alpha(W) \le 6k$, in time $2^{\OO(k^2)} n^{\OO(k)}$ either determines that the tree-independence number of $G$ is larger than $k$ or returns a rooted tree decomposition $\td$ of $G$ with independence number at most $8k$ such that $W$ is contained in the root bag of~$\td$.
\end{lemma}
\begin{proof}
The algorithm will be based on recursively constructing the decomposition, using $W$ as the interface in the recursion.
First, if $\alpha(G) \le 6k$, we return the trivial tree decomposition with only one bag $V(G)$.
Otherwise, we start by inserting arbitrary vertices of $G$ into $W$ until the condition $\alpha(W) = 6k$ holds.

Then, we apply the algorithm of \cref{lem:decomp_sepfind} to find a partition $(S, C_1, C_2, C_3)$ of $V(G)$ such that $S$ is a $(C_1, C_2, C_3)$-separator, $\alpha(S) \le 2k$, $\alpha(W \cap C_i) \le 4k$ for each $i \in \{1,2,3\}$, and at most one of $C_1$, $C_2$, $C_3$ is empty, or to determine that the tree-independence number of $G$ is larger than $k$, in this case returning ``no'' immediately.
Then, we construct the tree decomposition recursively as follows: for each $i \in \{1,2,3\}$, we recursively use the algorithm with the graph $G_i = G[C_i \cup S]$ and the set $W_i = (C_i \cap W) \cup S$.
Let $\td_1$, $\td_2$, $\td_3$ be the obtained tree decompositions and let $r_1$, $r_2$, $r_3$ be their root nodes.
We create a new root node $r$ with a bag $X_r = S \cup W$, and connect $r_1$, $r_2$, and $r_3$ as children of~$r$.

\cref{lem:decomp_sepfind} guarantees that $\alpha(C_i \cap W) \le 4k$ and $\alpha(S) \le 2k$, and therefore $\alpha(W_i) \le 6k$.
Also, $\alpha(S \cup W) \le 8k$ because $\alpha(W) \le 6k$ and $\alpha(S) \le 2k$.
Therefore, the  independence number of the constructed tree decomposition  is at most~$8k$.
The constructed  tree decomposition satisfies all conditions of tree decompositions:
Because $S$ is a separator between $C_1$, $C_2$, and $C_3$, when recursing into the graphs $G_i = G[C_i \cup S]$ for $i \in \{1,2,3\}$,  the union of the graphs $G_1$, $G_2$, and $G_3$ includes all vertices and edges of~$G$. Therefore, by induction, every vertex and edge will be contained in some bag of the constructed tree decomposition (the base case is $\alpha(G) \le 6k$).
By induction, the decomposition satisfies also the connectivity condition: if a vertex occurs in $G_i$ and $G_j$ for $i \neq j$, then it is in $S$ and therefore in the bag $X_r$ and also in the sets $W_i$ and $W_j$ and therefore in the root bags $X_{r_i}$ and $X_{r_j}$ of $\td_i$ and $\td_j$.

It remains to argue that the size of the recursion tree (and equivalently the size of the decomposition constructed) is $n^{\OO(1)}$.
First, by the guarantee of \cref{lem:decomp_sepfind} that $C_i$ is empty for at most one $i \in \{1,2,3\}$, we have that each $G_i$ has strictly fewer vertices than $G$, and therefore the constructed tree has height at most~$n$.
We say that a constructed node $t$ is a forget node if there exists a vertex $v$ contained in the bag of $t$, but not in the bag of the parent of $t$.
The number of forget nodes is at most $n$ because a vertex can be forgotten only once in a tree decomposition.

Recall that in the start of each recursive call, on a graph $G_i$ and a subset $W_i$, we either recognize that $\alpha(G_i) \le 6k$, creating a leaf node in this case, or add vertices to $W_i$ until $\alpha(W_i) = 6k$.
In the latter case, as these added vertices were not initially in $W_i$, they are not in the bag of the parent node, and therefore the node constructed in such a call will be a forget node if any such vertices are added.
Therefore, the new node constructed can be a non-forget non-leaf node only if $\alpha(W_i) = 6k$ already for the initial input~$W_i$.
Then, we observe that $\alpha(W_i) = 6k$ can hold for the initial input $W_i$ only if $\alpha(C_i \cap W) = 4k$ did hold for the corresponding component $C_i$ of the parent and the corresponding set~$W$.
Therefore, as $\alpha(C_i \cap W) = 4k$ can hold for at most one $i \in \{1,2,3\}$, we have that any node can have at most one non-forget non-leaf child node.

It follows that non-forget non-leaf nodes can be decomposed into maximal paths going between a node and its ancestor, and each of these paths has at most $n$ nodes by the height of the tree.
Each such path either starts at the root or its highest node is a child of a forget node.
Since the number of forget nodes is at most $n$, the number of maximal paths of non-forget non-leaf nodes is at most $n+1$, and therefore the number of non-forget non-leaf nodes is at most~$n(n+1)$.
The number of leaf nodes is at most three times the number of non-leaf nodes, so the total number of nodes is $\OO(n^2)$.

Therefore, the algorithm works by $\OO(n^2)$ applications of the algorithm of \cref{lem:decomp_sepfind}, and therefore its time complexity is $2^{\OO(k^2)} n^{\OO(k)}$.
\end{proof}

It remains to observe that by using iterative compression, we can satisfy the requirement of \cref{lem:maindecompalg} to have a tree decomposition with independence number $\OO(k)$ as an input (in particular, here the independence number will be at most $8k+1$), and therefore \cref{lem:maindecompalg} implies \cref{the:main_alg}.

In more detail, we order the vertices of $G$ as $v_1, \ldots, v_n$, and iteratively compute tree decompositions with independence number at most $8k$ for induced subgraphs $G[\{v_1, \ldots, v_i\}]$, for increasing values of~$i$.
The iterative computation guarantees that when computing the tree decomposition for $G[\{v_1, \ldots, v_i\}]$, we have the tree decomposition for $G[\{v_1, \ldots, v_{i-1}\}]$ with independence number at most $8k$ available, which can be used to obtain a tree decomposition with independence number at most $8k+1$ of $G[\{v_1, \ldots, v_i\}]$ by adding $v_i$ to each bag to be used as the input tree decomposition.
More precisely, we use \cref{lem:maindecompalg} to either determine in time $2^{\OO(k^2)} n^{\OO(k)}$ that the tree-independence number of $G[\{v_1, \ldots, v_i\}]$ is larger than $k$ or obtain a rooted tree decomposition $\td$ of $G[\{v_1, \ldots, v_i\}]$ with independence number at most~$8k$.
As the tree-independence number does not increase when taking induced subgraphs, if for some induced subgraph we conclude that the tree-independence number is larger than $k$, we can conclude the same holds also for~$G$.
Otherwise, after $n$ steps we will have a rooted tree decomposition $\td$ of $G$ with independence number at most~$8k$.

\section{Hardness of computing tree-independence number}
\label{sec:lower-bound}

In this section, we complement our main algorithmic result by complexity lower bounds.

The following folklore observation will be useful for us.

\begin{observation}\label{obs:folk}
Let $G$ be a graph and let $A,B\subseteq V(G)$ be disjoint subsets of vertices such that each vertex of $A$ is adjacent to every vertex of $B$, that is, the edges between $A$ and $B$ compose a biclique.
Then for every tree decomposition  $\mathcal{T} = (T, \{X_t\}_{t\in V(T)})$ of  $G$, there is $t\in V(T)$ such that $A\subseteq X_t$ or $B\subseteq X_t$.
\end{observation}

\begin{proof}
Let $G'$ be the graph obtained from $G$ by adding edges so that each bag $X_t$ becomes a clique.
Then $G'$ is a chordal supergraph of~$G$.
Note that either $A$ or $B$ is a clique in $G'$, since otherwise the two vertices of any non-edge in $A$ and two vertices of any non-edge in $B$ would form an induced $4$-cycle in~$G'$.
Since $\mathcal{T}$ is also a tree decomposition of $G'$ and each clique is covered by a single bag of the tree decomposition, there exists a node $t\in V(T)$ such that $A\subseteq X_t$ or $B\subseteq X_t$.
\end{proof}

In particular, we use this observation to show that from the point of view of parameterized approximation, approximating the tree-independence number of $G$ is not easier than approximating the independence number.

\begin{lemma}\label{lem:inappr}
Suppose there is an algorithm that, for a given $n$-vertex graph $G$ and a positive integer $k$, in time $t(k,n)$ can distinguish between the cases $\tin(G) \le k$ and $\tin(G) > f(k)$, where $f$ is a computable function.
Then, there is an algorithm that, for a given $n$-vertex graph $G$ and a positive integer $k$, in $t(k,2n) + n^{\OO(1)}$ time can distinguish between the cases $\alpha(G) \le k$ and $\alpha(G) > f(k)$.
\end{lemma}

\begin{proof}
We show the lemma by giving a polynomial-time algorithm that given an $n$-vertex graph $G$ constructs a $2n$-vertex graph $G'$ with $\tin(G') = \alpha(G)$.

We consider the reduction used by Dallard et al.\ in~\cite{DallardMS24} to prove that computing the tree-independence number is \NP-hard.
Construct the graph $G'$ by taking two disjoint copies $G_1$ and $G_2$ of $G$ and making each vertex of $G_1$ adjacent to every vertex of~$G_2$. Clearly, $\alpha(G')=\alpha(G)$, because a set of vertices $X$ is an independent set of $G'$ if and only if $X$ is either an independent set of $G_1$ or an independent set of~$G_2$.
By \cref{obs:folk}, we have that every tree decomposition of $G'$ has a bag that contains $V(G_1)$ or $V(G_2)$. Therefore, the trivial tree decomposition of $G'$ with the unique bag $V(G')$ is optimal and $\tin(G')=\alpha(G')=\alpha(G)$.
\end{proof}

The lemma can be used to obtain inapproximability lower bounds for computing the tree-independence number by using existing lower bounds for the approximation of the independence number.
We refer to the surveys of Feldmann et al.~\cite{FeldmannSLM20} and Ren~\cite{Ren21}, and the recent paper of Karthik and Khot~\cite{KarthikK21} for the statements of various complexity assumptions and lower bounds based on them which can be combined with \Cref{lem:inappr}.
Here, we spell out the consequences of the \W-hardness of independent set approximation by Lin~\cite{Lin21} and the {\sf Gap-ETH} result of Chalermsook et al.~\cite{DBLP:journals/siamcomp/ChalermsookCKLM20}.

\begin{theorem}\label{thm:inappr}
 For any constant $c\geq 1$, there is no algorithm running in  $f(k)\cdot n^{\OO(1)}$ time for a computable function $f(k)$ that, given an $n$-vertex graph and a positive integer $k$,
 can distinguish between the cases $\tin(G)\leq k$ and $\tin(G)>ck$, unless $\FPT=\W$.
 Moreover, assuming {\sf Gap-ETH}, for any computable function $g(k) \ge k$, there is no algorithm running in $f(k) \cdot n^{o(k)}$ time for a computable function $f(k)$ that, given an $n$-vertex graph and a positive integer $k$,
 can distinguish between the cases $\tin(G)\leq k$ and $\tin(G)>g(k)$.
 \end{theorem}

The simple reduction from the proof of \Cref{lem:inappr} immediately implies that it is \W-hard to decide whether $\tin(G)\leq k$ for the parameterization by~$k$.
However, the problem is, in fact, harder. We prove that it is \NP-complete to decide whether $\tin(G)\leq 4$.
For this, we use the following auxiliary result.
We use $\tau(G)$ to denote the \emph{vertex cover} number of $G$, that is, the minimum size of a set $S\subseteq V(G)$ such that for every edge $uv\in E(G)$, $u\in S$ or $v\in S$.

\begin{lemma}\label{lem:aux}
    It is \NP-complete to decide whether a given graph $G$ has two disjoint subsets $U_1,U_2\subseteq V(G)$ such that, for $i = 1,2$,
    \begin{itemize}
    \item[(i)] $\tau(G[U_i])\leq 1$ and
    \item[(ii)] for every clique $K$ of size three, $U_i\cap K\neq \emptyset$.
    \end{itemize}
    Moreover, the problem remains \NP-complete when restricted to graphs without cliques of size four.
\end{lemma}

\begin{proof}
We reduce from the $3$-\textsc{Coloring} problem. Recall that the task of $3$-\textsc{Coloring} is to decide whether a graph $G$ admits a \emph{proper $3$-coloring}, that is, its vertices can be colored by three colors in such a way that adjacent vertices receive distinct colors. Equivalently,
the vertex set of $G$ has a partition $(I_1,I_2,I_3)$ into independent sets, called \emph{color classes}.
The problem is well-known to be \NP-complete~\cite{GareyJ79}. Note that the problem stays \NP-complete for graphs that
have no cliques of size four.
Observe also that $3$-\textsc{Coloring} is \NP-complete for graphs $G$ such that each edge is contained in a triangle, that is, a cycle of length three.
To see this, let $G'$ be the graph obtained from $G$ by constructing a new vertex $w_{uv}$ for every edge $uv\in E(G)$ and making $w_{uv}$ adjacent to both $u$ and~$v$. We have that every edge of $G'$ is in a triangle and it is easy to see that $G$ has a proper $3$-coloring if and only if $G'$ has the same property.

For an integer $k\geq 3$, the \emph{wheel} graph $W_k$ is the graph obtained from a cycle $C_k$ on $k$ vertices by adding a new vertex and making it adjacent to each vertex of the cycle.

Let $G$ be a graph that does not contain cliques of size four such that every edge is in a triangle. We construct $G'$ by taking the disjoint union of $G$ and two copies of $W_5$ denoted by $H_1$ and~$H_2$. Notice that $G'$ has no cliques of size four.
We claim that $V(G)$ has a partition into three color classes if and only if $G'$ has two disjoint subsets $U_1,U_2\subseteq V(G')$ satisfying conditions (i) and (ii) of the lemma.

Suppose that $(I_1,I_2,I_3)$ is a partition of $V(G)$ into three color classes. We set $U_1:=I_1$, $U_2:=I_2$, and add to the sets some vertices of $H_1$ and $H_2$ as it is shown in \Cref{fig:H}.
By construction, $\tau(G'[U_i])\leq 1$ for $i=1,2$.
Suppose that $K$ is a clique in $G'$ of size three.
If $K$ is a clique in $H_1$ or $H_2$, then $U_i\cap K\neq \emptyset$ for $i=1,2$ by construction (see \Cref{fig:H}). If $K$ is a clique in $G$, then the vertices of $K$ are colored by distinct colors. Therefore, $U_i\cap K\neq \emptyset$ for $i=1,2$.

\begin{figure}[ht]
\centering
\scalebox{0.7}{
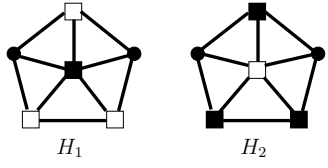
}
\caption{The placement of the vertices of $H_1$ and $H_2$ in $U_1$ and $U_2$; the vertices denoted by black squares are in $U_1$ and the vertices denoted by white squares are in~$U_2$.}\label{fig:H}
\end{figure}

Suppose now that $G'$ has two disjoint subsets $U_1,U_2\subseteq V(G')$ satisfying conditions (i) and (ii). Let $W=V(G')\setminus (U_1\cup U_2)$. Recall that every edge of $G$ and, therefore, every edge of $G'$ is contained in a triangle. Hence, if $uv$ is an edge of $G'$, then there is a clique $K$ of size three such that $u,v\in K$. We have that $U_i\cap K\neq\emptyset$ for $i=1,2$. This means that at most one vertex of $K$ is in~$W$. Hence, either $u\notin W$ or $v\notin W$ implying that
 $W$ is an independent set.  Because $W_5$ is not $3$-colorable and $W$ is an independent set, we have that $H_1$ contains at least two adjacent vertices of $U_1$ or $U_2$ and the same holds for~$H_2$. Because $\tau(G[U_i])\leq 1$ for $i=1,2$, we have that  $\tau(G[U_i\cap (V(H_1)\cup V(H_2))])=1$ for $i=1,2$. Therefore, $I_1=V(G)\cap U_1$ and $I_2=V(G)\cap U_2$ are independent sets and  $(I_1,I_2,I_3)$, where $I_3=W\cap V(G)$, is a partition of $V(G)$ into three color classes.
\end{proof}

We use \Cref{lem:aux} to show \Cref{the:main_nphard} which we restate here.

\npcmain*

\begin{proof}
We show the theorem for $k=4$ and then explain how to extend the proof for $k\geq 5$.

\begin{figure}[ht]
\centering
\scalebox{0.7}{
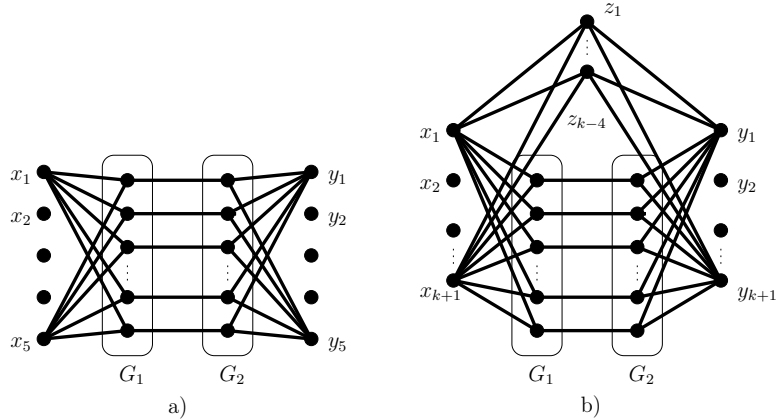
}
\caption{The construction of $G'$ for $k=4$ (a) and for $k\geq 5$ (b).}\label{fig:GNPc}
\end{figure}

We reduce from the problem from \Cref{lem:aux}.
Let $G$ be a graph without cliques of size four. We construct the graph $G'$ as follows (see \Cref{fig:GNPc} (a)).
\begin{itemize}
\item Take two disjoint copies $G_1$ and $G_2$ of $\overline{G}$ and make every vertex of $G_1$ adjacent to its copy in~$G_2$.
\item Add 5 new vertices $x_1,\ldots,x_5$ and make each of them adjacent to every vertex of~$G_1$.
 \item Add 5 new vertices $y_1,\ldots,y_5$ and make each of them adjacent to every vertex of~$G_2$.
\end{itemize}
Throughout the proof we denote for any vertex $v\in V(\overline{G}) = V(G)$ the two copies of $v$ in $G_1$ and $G_2$ by $v^{(1)}$ and $v^{(2)}$, respectively.

It is easy to see that the graph $G'$ can be constructed in polynomial time.
We claim that $G$ has two disjoint subsets $U_1,U_2\subseteq V(G)$ such that  (i) $\tau(G[U_i])\leq 1$ for $i=1,2$
and (ii) for every clique $K$ of size three, $U_i\cap K\neq \emptyset$ for $i=1,2$ if and only if $\tin(G')\leq 4$.

Suppose that  $G$ has two disjoint subsets $U_1,U_2\subseteq V(G)$ satisfying (i) and (ii).  We construct a tree decomposition  $\mathcal{T} = (T, \{X_t\}_{t\in V(T)})$ of  $G'$ such that $\alpha(G'[X_t])\leq 4$ for all $t\in V(T)$.
Denote by $v_1,\ldots,v_n$ the vertices of $G$ and assume that (a) $U_1=\{v_1,\ldots,v_\ell\}$ and $U_2=\{v_r,\ldots,v_n\}$ for $1\leq\ell<r\leq n$, and (b) $\{v_1,\ldots,v_{\ell-1}\}$ and $\{v_{r+1},\ldots,v_n\}$ are independent sets in $G$; we can make assumption (b) because $\tau(G[U_i])\leq 1$ for $i=1,2$.
Then $\mathcal{T}$ is constructed as follows.
\begin{itemize}
\item To construct $T$, introduce a path $P=t_0\cdots t_{2n}$ with $2n+1$ nodes, and then set the corresponding bags  $X_{t_{2i}}=\{v_{i+1}^{(1)},\ldots,v_n^{(1)}\}\cup\{v_1^{(2)},\ldots,v_{i}^{(2)}\}$ for $i\in\{0,\ldots,n\}$ and
$X_{t_{2i-1}}=\{v_i^{(1)},\ldots,v_n^{(1)}\}\cup\{v_1^{(2)},\ldots,v_{i}^{(2)}\}$ for $i\in\{1,\ldots,n\}$; in particular $X_{t_0}=V(G_1)$ and $X_{t_{2n}}=V(G_2)$.
\item For every $j\in\{1,\ldots,5\}$, construct a node $t_j'$ of $T$, make it adjacent to $t_0$, and set $X_{t_j'}=V(G_1)\cup\{x_j\}$.
\item For every $j\in\{1,\ldots,5\}$, construct a node $t_j''$ of $T$, make it adjacent to $t_{2n}$, and set $X_{t_j''}=V(G_2)\cup\{y_j\}$.
\end{itemize}
The construction immediately implies that $\mathcal{T}$ is indeed a  feasible tree decomposition of~$G'$. We claim that $\alpha(G'[X_t])\leq 4$ for every node $t$ of~$T$.

Because $G$ has no cliques of size four, we have $\alpha(G_1)=\alpha(G_2)=\alpha(\overline{G})\leq 3$.
Therefore, $\alpha(G'[X_t])\leq 3$ for $t\in\{t_1',\ldots,t_5'\}\cup\{t_1'',\ldots,t_5''\}\cup\{t_0,t_{2n}\}$. 
Notice that for every $i\in\{1,\ldots,n\}$, $X_{t_{2i}}\subset X_{t_{2i-1}}$.
Hence, it is sufficient to prove that $\alpha(G'[X_{t_{2i-1}}])\leq 4$ for every $i\in\{1,\ldots,n\}$.
Let $i\in\{1,\ldots,n\}$ and let $I$ be an independent set of maximum size in $X_{2i-1}=\{v_i^{(1)},\ldots,v_n^{(1)}\}\cup\{v_1^{(2)},\ldots,v_{i}^{(2)}\}$.
We consider the following cases.

Suppose that $i<\ell$. We have that $\{v_1,\ldots,v_i\}$ is an independent set in~$G$. This means that  $\{v_1^{(2)},\ldots,v_{i}^{(2)}\}$ is a clique in $G_2$ and $|I\cap \{v_1^{(2)},\ldots,v_{i}^{(2)}\}|\leq 1$.  Because $G$ has no cliques of size four, $|I\cap \{v_i^{(1)},\ldots,v_n^{(1)}\}|\leq \alpha(G_1)\leq 3$ and we conclude that $|I|\leq 4$. Notice that the case $i>r$ is symmetric and we have that $|I|\leq 4$ by the same arguments.

Assume that $i=\ell$. Then  $\{v_1,\ldots,v_i\}=U_1$. If $v_i^{(2)}\notin I$, then $|I\cap \{v_1^{(2)},\ldots,v_{i}^{(2)}\}|\leq 1$, because $\{v_1^{(2)},\ldots,v_{i-1}^{(2)}\}$ is a clique.
Then similarly to the previous case, we observe that $|I\cap \{v_i^{(1)},\ldots,v_n^{(1)}\}|\leq \alpha(G_1)\leq 3$ and conclude that $|I|\leq 4$. Suppose that $v_i^{(2)}\in I$. Then we have that $|I\cap \{v_1^{(2)},\ldots,v_{i}^{(2)}\}|\leq 2$, because $\tau(G[\{v_1,\ldots,v_i\}])\leq 1$.
Since $v_i^{(2)}\in I$ and $v_i^{(1)}v_i^{(2)}\in E(G')$, we infer that $v_i^{(1)}\notin I$.
Therefore, $I\cap \{v_i^{(1)},\ldots,v_n^{(1)}\}\subseteq \{v_{i+1}^{(1)},\ldots,v_n^{(1)}\}$. Recall that for any clique $K$ of size three, $K\cap U_1\neq\emptyset$, that is, $G\setminus U_1$ has no cliques of size three.
Equivalently, $\alpha(G_1\setminus U_1)\leq 2$.
Thus, $|I\cap  \{v_i^{(1)},\ldots,v_n^{(1)}\}|=|I\cap \{v_{i+1}^{(1)},\ldots,v_n^{(1)}\}|\leq
\alpha(G_1[\{v_{i+1}^{(1)},\ldots,v_n^{(1)}\}])=\alpha(G\setminus U_1)\leq 2$.
This implies that  $|I|\leq 4$.
The case $i=r$ is symmetric and $|I|\leq 4$ by the same arguments.

Finally, we assume that $\ell<i<r$. Clearly, $I\cap\{v_1^{(2)},\ldots,v_i^{(2)}\}\subseteq \{v_1^{(2)},\ldots,v_{r-1}^{(2)}\}$. Since $G\setminus U_2$ has no clique of size three,
$\alpha(G_2[\{v_1^{(2)},\ldots,v_{r-1}^{(2)}\}])\leq 2$. Thus, $|I\cap\{v_1^{(2)},\ldots,v_i^{(2)}\}|\leq 2$. By symmetry, we also have that $|I\cap\{v_i^{(1)},\ldots,v_n^{(1)}\}|\leq 2$.
Hence, $|I|\leq 4$. This concludes the case analysis.

We obtain that $\alpha(G'[X_t])\leq 4$ for every $t\in V(T)$. Therefore, $\tin(G')\leq 4$.

For the opposite direction, assume that $\tin(G')\leq 4$. Consider a tree decomposition $\mathcal{T} = (T, \{X_t\}_{t\in V(T)})$ of  $G'$ such that $\alpha(G'[X_t])\leq 4$ for every $t\in V(T)$.
For every $j\in\{1,\ldots,5\}$, vertex $x_j$ is adjacent to every vertex of~$G_1$. By \cref{obs:folk}, there is a node $t'\in V(T)$ such that $\{x_1,\ldots,x_5\}\subseteq X_{t'}$ or $V(G_1)\subseteq X_{t'}$.
However, $\{x_1,\ldots,x_5\}$ is an independent set of size five and no bag can contain all these vertices.
Thus, $V(G_1)\subseteq X_{t'}$.
By symmetry, there is a node $t''\in V(T)$ such that   $V(G_2)\subseteq X_{t''}$.
We assume without loss of generality that $X_{t'}=V(G_1)$ and $X_{t''}=V(G_2)$. Otherwise, if, say, $V(G_1)\subset X_{t'}$, we can add a leaf node to $T$, make it adjacent to $t'$, and assign the bag $V(G_1)$ to the new node.
Consider the $t',t''$-path $P$ in $T$ and set $Y_t=X_t\cap(V(G_1)\cup V(G_2))$ for $t\in V(P)$. We claim that
$\mathcal{P}=(P,\{Y_t\}_{t\in V(P)})$ is a path decomposition of $H=G'[V(G_1)\cup V(G_2)]$.

Because $V(G_1)=Y_{t'}$ and $V(G_2)=Y_{t''}$, every vertex of $H$ is included in some bag. Since $P$ is a path in $T$, for every vertex $v\in V(H)$, the subgraph of $P$ induced by $\{t\in V(P)\colon v\in Y_t\}$ is a subpath of~$P$. Let $uv$ be an edge of~$H$. If $uv\in E(G_1)$, then $u,v\in Y_{t'}$, and if $uv\in E(G_2)$, then $u,v\in Y_{t''}$. Suppose that $u\in V(G_1)$ and $v\in V(G_2)$. Then $u$ and $v$ are two copies of the same vertex of~$G$. Because $\mathcal{T}$ is a tree decomposition of $G'$, there is a node $t\in V(T)$  such that $u,v\in X_t$.
Note that $u\in Y_{t'}$, $v\notin Y_{t'}$, $v\in Y_{t''}$, and $u\notin Y_{t''}$.
Then either $t$ is an internal vertex of the path $P$, or the shortest path in $T$ between $t$ and $P$ has its end-vertex $t^*$ in an internal vertex of~$P$.
In the first case, $u,v\in Y_t$, and $u,v\in Y_{t^*}$ in the second.
This completes the proof of our claim.

Using standard arguments (see, e.g., the textbook~\cite{cygan2015parameterized}, or~\cite{DallardMS24} for a treatment focused on the independence number), we can assume that $\mathcal{P}$ is \emph{nice}, that is, $|Y_{s}\bigtriangleup Y_{s'}|\leq 1$ for every two adjacent nodes $s$ and $s'$ of~$P$.
Recall that $Y_{t'}\cap V(G_2)=\emptyset$ and $Y_{t''}=V(G_2)$.
Therefore, there are distinct nodes $t_1,\ldots,t_n\in V(P)$ sorted along the path order with respect to $P$ such that the vertices $v_1^{(2)},\ldots,v_n^{(2)}$ are \emph{introduced} in $t_1,\ldots,t_n$, that is, for $i\in \{1,\ldots,n\}$, $v_i^{(2)}$ is contained in the bags $Y_t$ for every $t$ in the $(t_i,t'')$-subpath of $P$ and $v_i^{(2)}\notin Y_t$ for every $t\neq t_i$ in the $(t',t_i)$-subpath.
Notice that because $v_i^{(2)}\in X_{t_i}$ and $v_i^{(2)}\in X_{t''}$, $v_i^{(2)}$ is included in every bag $X_t$, where  $t$ is in the $(t_i,t'')$-subpath of~$P$.
Hence, $v_1^{(2)},\ldots,v_i^{(2)}\in Y_{t_i}$ for all $i\in\{1,\ldots,n\}$.
Recall that for every $i\in\{1,\ldots,n\}$, the vertex  $v_i^{(2)}$ is the copy in $G_2$ of a vertex $v_i$ of $\overline{G}$, and we denote the corresponding copy of $v_i$ in $G_1$ by $v_i^{(1)}$.
Let $i\in\{1,\ldots,n\}$.
Because $\mathcal{P}$ is a path decomposition of $H$ and $v_i^{(1)}v_i^{(2)}\in E(H)$, there is $t\in V(P)$ such that $v_i^{(1)},v_i^{(2)}\in X_t$. By the definition of $t_i$, $t\geq t_i$. Since $v_i^{(1)}\in X_{t'}$, $v_i^{(1)}$ is included in each bag for nodes of $P$ in the $(t',t_i)$-subpath.
This implies that $v_i^{(1)},\ldots,v_n^{(1)}\in Y_{t_i}$ for each $i\in\{1,\ldots,n\}$.
For every $i\in\{1,\ldots,n\}$, let $Z_i=\{v_i^{(1)},\ldots,v_n^{(1)}\}\cup\{v_1^{(2)},\ldots,v_i^{(2)}\}$.
We obtain that $Z_i\subseteq Y_{t_i}$ for each $i\in\{1,\ldots,n\}$.
In particular, this implies that
$\alpha(G'[Z_i])\leq 4$ for each $i\in\{1,\ldots,n\}$.

We select minimum $\ell\in\{1,\ldots,n-1\}$ such that $\tau(G[\{v_1,\ldots,v_\ell\}])=1$ if such an $\ell$ exists and we set $\ell=n-1$ if $\{v_1,\ldots,v_{n-1}\}$ is an independent set of~$G$.
We choose maximum $r\in\{\ell+1,\ldots,n\}$ such that $\tau(G[\{v_r,\ldots,v_n\}])=1$ if such an $r$ exists and $r=\ell+1$ if $\{v_{\ell+1},\ldots,v_n\}$ is an independent set of~$G$. We define $U_1=\{v_1,\ldots,v_\ell\}$ and
$U_2=\{v_r,\ldots,v_n\}$. Observe that $\tau(G[U_1])\leq 1$ and $\tau(G[U_2])\leq 1$ by construction. We claim that for every clique $K$ of $G$ of size three, $U_i\cap K\neq \emptyset$ for $i=1,2$.

Suppose that $\tau(G[U_1])=1$. Then $\alpha(G'[\{v_1^{(2)},\ldots,v_\ell^{(2)}\}])=2$. Because $\alpha(G'[Z_\ell])\leq 4$, we have that $\alpha(G'[\{v_{\ell+1}^{(1)},\ldots,v_n^{(1)}\}])\leq 2$, that is, the graph  $G_1[\{v_{\ell+1}^{(1)},\ldots,v_n^{(1)}\}]$ does not have an independent set of size three.
Then $\overline{G}[\{v_{\ell+1},\ldots,v_n\}]$ has the same property.
Thus, $U_1$ intersects any clique of size three in~$G$.
Suppose that $U_1$ is an independent set in $G$.
Then $U_1=\{v_1,\ldots,v_{n-1}\}$. 
Trivially, $U_1$ intersects any clique of size three in~$G$. 
If $\tau(G[U_2])=1$, then we have that $U_2$ intersects any clique of size three in $G$ by the same arguments as for $U_1$ by symmetry. 
Suppose that $U_2$ is an independent set in $G$. 
Then $r=\ell+1$ and $V(G)\setminus U_2=U_1$. 
Because $\tau(G[U_1])\leq 1$, the graph $G[U_1]$ has no clique of size three.
We conclude that $U_2$ intersects every clique of size three in~$G$.

This concludes the proof of the theorem for~$k=4$. To prove the statement for $k\geq 5$, we slightly modify the reduction as it is shown in \Cref{fig:GNPc} (b). In the construction of $G'$,  instead of adding five new vertices $x_1,\ldots,x_5$ and five vertices $y_1,\ldots,y_5$, we create $k+1$ vertices $x_1,\ldots,x_{k+1}$ and $k+1$ vertices $y_1,\ldots,y_{k+1}$.
Furthermore, we create $k-4$ new vertices $z_1,\ldots,z_{k-4}$ and make them adjacent to  $x_1,\ldots,x_{k+1}$ and $y_1,\ldots,y_{k+1}$. The other parts of the construction remain the same. Then we show that $G$ has two disjoint subsets $U_1,U_2\subseteq V(G)$ such that  (i) $\tau(G[U_i])\leq 1$ for $i=1,2$
and (ii) for every clique $K$ of size three, $U_i\cap K\neq \emptyset$ for $i=1,2$ if and only if $\tin(G')\leq k$ using almost the same arguments as for the case~$k=4$.

For the forward direction, in the construction of the tree decomposition of $G'$, we construct $k+1$ nodes $t_j'$ and $k+1$ nodes $t_j''$ instead of five and
include $z_1,\ldots,z_{k-4}$ in every bag of the decomposition. Then we obtain a tree decomposition whose bags induce subgraphs with independence number at most~$k$.
For the opposite direction, if there is a tree decomposition $\mathcal{T} = (T, \{X_t\}_{t\in V(T)})$ of  $G'$ such that $\alpha(G'[X_t])\leq k$ for every $t\in V(T)$, we observe that there are $t',t''\in V(T)$ such that
$V(G_1)\cup\{z_1,\dots,z_{k-4}\}\subseteq X_{t'}$ and $V(G_2)\cup\{z_1,\dots,z_{k-4}\}\subseteq X_{t''}$. Then for every $t\in V(P)$, where $P$ is the $(t',t'')$-path in $T$, $z_1,\ldots,z_{k-4}\in X_t$. This allows us to use the same arguments as in the case $k=4$ to show that $G$ has two disjoint subsets $U_1,U_2\subseteq V(G)$ satisfying (i) and (ii). This concludes the proof.
\end{proof}

\section{Hardness of finding a separator with bounded independence number\label{sec:hardsep}}

In this section we show that, for every fixed integer $k\ge 3$, deciding if two given vertices of a graph can be separated by removing a set of vertices that induces a graph with independence number at most $k$ is \NP-complete.
To put this result in perspective, note that the case with $k = 1$ is polynomial since we can compute all clique cutsets in polynomial time using Tarjan's algorithm~\cite{MR798539}.
The case with~$k=2$ is still open.

A graph property is \emph{nontrivial} if there is at least one graph having the property and there is at least one graph that does not have the property.
We say that a graph property is \emph{additive hereditary} if it is closed under taking vertex-disjoint unions and induced subgraphs.
For a graph $G$ and any two nontrivial additive hereditary graph properties $\mathcal{P}$ and $\mathcal{Q}$, we say that $G$ is \emph{$(\mathcal{P},\mathcal{Q})$-colorable} if $V(G)$ can be partitioned into two sets $A$ and $B$ such that $G[A]$ has property $\mathcal{P}$ and $G[B]$ has property $\mathcal{Q}$.
In \cite{MR2097312}, Farrugia showed that deciding if a given graph is $(\mathcal{P},\mathcal{Q})$-colorable is \NP-hard, except if both $\mathcal{P}$ and $\mathcal{Q}$ are the property of being edgeless (in which case the problem corresponds to deciding $2$-colorability).

\begin{theorem}\label{thm:separator}
For every integer $k\ge 3$, it is \NP-complete to decide, given a graph $H$ and two distinct vertices $u,v \in V(H)$, if there exists a $u{,}v$-separator $S$ such that $\alpha(H[S]) \leq k$.
\end{theorem}

\begin{proof}
Let $G$ be a graph with vertices $v_1, \dots, v_n$.

For two distinct integers $i,j \in \{1, \dots, n\}$, let $J = J_{i,j}(G)$
be the graph obtained as follows (see \Cref{fig:GHNPc} (a)):
\begin{itemize}
    \item take two disjoint copies $G_L$ and $G_R$ of $\overline{G}$, with vertex sets $V(G_L) = \{v_1^L, \dots, v_n^L\}$ and $V(G_R) = \{v_1^R, \dots, v_n^R\}$, respectively, such that for each
    $p\in \{1, \dots, n\}$, the vertices $v_p^L$ and $v_p^R$ both correspond to the vertex $v_p$, and add, for each vertex $v_p \in V(G)$, the edge $v_p^L v_p^R$; we call such edges the \emph{middle edges} of $J$;
    \item create two disjoint sets $Z_L$ and $Z_R$ of four new vertices each;
    \item finally, connect $v_i^L$ by edges to all the vertices in $Z_R$, and $v_j^R$ to all the vertices in $Z_L$.
\end{itemize}
Any such graph $J_{i,j}(G)$ is called a \emph{$J$-graph} of~$G$.
The vertices in $V(G_L)\cup Z_L$ and $V(G_R)\cup Z_R$ are referred to as the \emph{left} and \emph{right vertices} of $J$, respectively.
Furthermore, the vertices in $V(G_L)\cup V(G_R)$ are referred to as the \emph{$G$-vertices} of $J$, while the vertices in $Z_L\cup Z_R$ are referred to as the \emph{$Z$-vertices} of~$J$.
Analogously, the vertices in $V(G_L)$ are referred to as the \emph{left $G$-vertices} of $J$, etc.

\begin{figure}[ht]
\centering
\scalebox{0.7}{
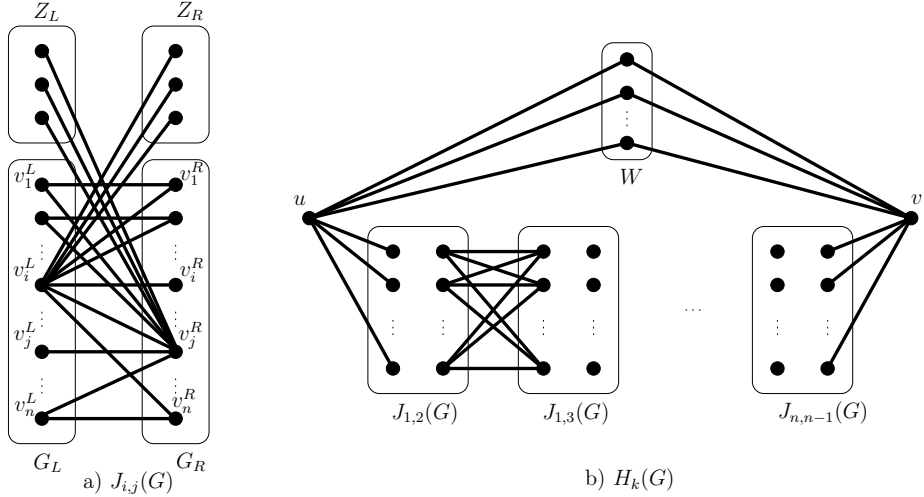
}
\caption{The construction of $J_{i,j}$ (a) and $H_k(G)$ (b).}\label{fig:GHNPc}
\end{figure}

We are now ready to describe the graph $H = H_k(G)$, which we construct as follows (see \Cref{fig:GHNPc} (b)):
\begin{itemize}
    \item first, take the disjoint union of the graphs $J_{i,j}(G)$ over all possible distinct values of $i,j \in \{1, \dots, n\}$;
    \item fix a total order on the $J$-graphs and add all possible edges between the right vertices of a $J$-graph and the left vertices of its successor in the order (if any);
    \item add two new vertices $u$ and $v$ and connect $u$ to the left vertices of the first $J$-graph and $v$ to the right vertices of the last $J$-graph;
    \item finally, add a set $W$ of $k-3$ new vertices and connect $u$ and $v$ to all the vertices in~$W$.
\end{itemize}

Let $\mathcal{P}$ be the property of being $K_2$-free and $\mathcal{Q}$ the property of being $K_3$-free.
Since the properties $\mathcal{P}$ and $\mathcal{Q}$ can be recognized in polynomial time, the problem of deciding if $G$ is $(\mathcal{P},\mathcal{Q})$-colorable is in \NP.
Furthermore, the problem is \NP-complete, since both $\mathcal{P}$ and $\mathcal{Q}$ are additive hereditary properties and the aforementioned result of Farrugia \cite{MR2097312} applies.
Note also that, since $k$ is fixed, deciding if there exists a $u{,}v$-separator $S$ in $H$ such that $\alpha(H[S]) \leq k$ is in \NP.

We claim that $G$ is $(\mathcal{P},\mathcal{Q})$-colorable if and only if $H = H_k(G)$ admits a $u{,}v$-separator $S$ with $\alpha(H[S]) \leq k$.
We first assume that $G$ is $(\mathcal{P},\mathcal{Q})$-colorable.
Let $(A,B)$ be a partition of $V(G)$ such that $G[A]$ is $K_2$-free and $G[B]$ is $K_3$-free.
There are two cases to consider, depending on whether one of $A$ or $B$ is empty, or both are nonempty.
Suppose first that $A$ is empty.
Then $B= V(G)$ and thus $G$ is $K_3$-free.
Fix one of the $J$-graphs of $G$, say $J = J_{i,j}(G)$, and let $S$ be the union of the set $L$ of left $G$-vertices of $J$ together with $W$ and $\{v_j^R\}$.
Observe that $S$ is a $u{,}v$-separator in $H$, since $S$ contains all the vertices in $W$ and any $u{,}v$-path in $G \setminus W$ that does not contain any vertex in $L$ has to contain $v_j^R$.
Recall that $L$ induces the complement of $G$ in $H$ and that $G$ is $K_3$-free.
Hence, we obtain that $\alpha(H[S]) = \alpha(H[W]) + \alpha(H[L \cup \{v_j^R\}]) \leq k-3+2+1 = k$.
The case when $B$ is empty is similar, except that this time $G$ is edgeless, hence, $H[L]$ is a complete graph and the same $S$ as above works.
So we may assume that neither $A$ nor $B$ is empty.
Let $v_i \in A$ and $v_j \in B$ and consider the $J$-graph $J = J_{i,j}(G)$.
We abuse notation and, for a set $X \subseteq V(G)$, denote by $X \cap V(J)$ the set of copies of vertices of $X$ in $V(J)$.
Let $L$ and $R$ be the sets of left and right vertices of $J$, respectively, and define $S = W \cup (L \cap A) \cup (R \cap B)$.
Observe that $S$ is a $u{,}v$-separator in $H$, since it contains all the vertices in $W$, intersects each middle edge of $J$ in exactly one endpoint, and contains both $v_i^L$ and $v_j^R$.
Furthermore, for the same reasons as above, we have $\alpha(H[S]) = \alpha(H[W]) + \alpha(H[L \cap A]) + \alpha(H[R \cap B]) \leq k-3 + 1 + 2 = k$.

Now, assume that $H$ admits a $u{,}v$-separator $S$ with $\alpha(H[S]) \leq k$.
We may assume without loss of generality that $S$ is inclusion-minimal.
Notice that $W = N(u) \cap N(v)$, and hence $W \subseteq S$, necessarily.
If each $J$-graph of $G$, say $J = J_{i,j}(G)$, contains an edge $e_{i,j}$ connecting a left vertex of $J$ with a right vertex of $J$, such that no endpoint of $e_{i,j}$ belongs to $S$, then choosing arbitrarily one such edge in each $J$-graph would result in a $u{,}v$-path in $G\setminus S$, a contradiction.
Therefore, there exists some $J$-graph of $G$, say $J = J_{i,j}(G)$, such that $S$ contains an endpoint of each edge connecting a left vertex of $J$ with a right vertex of~$J$.
Since $S\cap (W\cup V(J))$ already separates $u$ from $v$, the minimality of $S$ implies that $J_{i,j}(G)$ is the only $J$-graph of $G$ containing vertices from $S$, that is, $S \subseteq V(J)\cup W$.
In what follows, we use the same notations for the sets of vertices as in the definition of $J_{i,j}(G)$ above.
We show next that $S$ contains both $v_i^L$ and $v_j^R$.
Suppose for a contradiction that $S$ does not contain $v_i^L$ or $v_j^R$; without loss of generality, we may assume that $v_j^R \notin S$.
Since $v_j^R$ is adjacent to all the vertices in $Z_L$ and $S$ contains an endpoint of each edge connecting a left vertex of $J$ with a right vertex of $J$, we have $Z_L \subseteq S$ and hence
$Z_L \cup W\subseteq S$.
However, $Z_L \cup W$ is an independent set in $H$, which implies that $\alpha(H[S]) \geq |Z_L|+|W| = 4 + k-3 = k+1$, a contradiction.
This shows that $S$ contains both $v_i^L$ and $v_j^R$, as claimed.

Let $L$ and $R$ be the sets of left and right vertices of $J$, respectively, and let $A,B\subseteq V(G)$ be the sets of vertices of $G$ whose copies in $J$ belong to $S \cap L$ and $S \cap R$, respectively.
Recall that $S$ contains both $v_i^L$ and $v_j^R$, and hence $v_i \in A$ and $v_j \in B$.
Since $S$ is a minimal \hbox{$u{,}v$-separator} and $\{v_i^L,v_j^R\}\subseteq S\cap V(J)$, the set $S\cap V(J)$ does not contain any vertex from $Z_L\cup Z_R$ and intersects every middle edge of $J$ in exactly one endpoint.
Thus, $(A,B)$ is a partition of $V(G)$.
We show that $(A,B)$ is a certificate to the $(\mathcal{P},\mathcal{Q})$-colorability of~$G$.
By assumption, $\alpha(H[S]) \leq k$, and we notice that $k-3$ of the vertices of any maximum independent set in $H[S]$ come from the vertices in~$W$.
Since $(A,B)$ is a partition of $V(G)$, there is no edge in $H$ between a vertex of $S \cap L$ and $S \cap R$.
This implies that $k-3 + \alpha(\overline{G}[A]) + \alpha(\overline{G}[B]) =\alpha(H[W]) + \alpha(H[S \cap L]) + \alpha(H[(S \cap R)]) = \alpha(H[S])\leq k$, and in particular that $\alpha(\overline{G}[A]) + \alpha(\overline{G}[B]) \leq 3$.
We may assume without loss of generality that
$\alpha(\overline{G}[A]) \leq \alpha(\overline{G}[B])$.
Since $A$ and $B$ are nonempty, we deduce that $\alpha(\overline{G}[A]) \leq 1$ and that $\alpha(\overline{G}[B]) \leq 2$.
This means that $G[A]$ is a $K_2$-free graph and $G[B]$ a $K_3$-free graph.
Thus, $G$ is $(\mathcal{P},\mathcal{Q})$-colorable, which concludes the proof.
\end{proof}

\section{Conclusion}\label{sec:concl}

The main result of our paper is an algorithm that, given an $n$-vertex graph $G$ and an integer $k$, in time $2^{\OO(k^2)} n^{\OO(k)}$ either outputs a tree decomposition of $G$ with independence number at most $8k$, or concludes that the tree-independence number of $G$ is larger than~$k$.
This yields also the same result for computing the minor-matching hypertree-width of a graph~\cite{Yolov18}.
Our results allow us to solve in $2^{\OO(k^2)} n^{\OO(k)}$ time a plethora of problems when the inputs are restricted to graphs of tree-independence number $k$ or minor-matching hypertree-width $k$~\cite{DallardMS24,LMMORS24,Yolov18}.
We now show that this result is tight in several aspects.

First, one could ask: What is the most general width-parameter defined by a min-max formula over the bags of a tree decomposition (see, e.g.~\cite{adler2006width,MR3144912}) that allows us to solve problems like \textsc{Maximum Independent Set} in polynomial time when bounded?
For parameters where the width of a bag depends only on the induced subgraph of the bag, this turns out to be $\tin$.
In particular, we recall that \textsc{Maximum Independent Set} is \NP-hard on graphs with each edge subdivided twice, but such graphs admit a tree decomposition where one bag is a large independent set, and the induced subgraphs of the other bags are isomorphic to $4$-vertex paths.
It follows that if the width-measure of a bag is monotone, i.e., it does not increase when taking induced subgraphs, it must be unbounded whenever $\alpha$ is unbounded.
In other words, if there would be a width parameter $\trlambda$ defined as the minimum, over all tree decomposition, of the maximum of $\lambda(G[X_t])$ over the bags $X_t$ of the tree decomposition, where $\lambda$ is a monotone graph invariant, then either \textsc{Maximum Independent Set} is already \NP-hard when $\lambda$ is a constant, or the parameterization by $\tin$ is more general than the parameterization by $\trlambda$.

The width parameter $\trmu$, the
minor-matching hypertree-width, escapes this argument because it does not only depend on the subgraphs induced by the bags, but also on the neighborhoods of the bags.
In particular, for $\trmu$ the width of a bag $X_t$ is defined as the maximum cardinality of an induced matching in $G$ whose every edge intersects~$X_t$.
A similar example shows that this type of parameters where the width of a bag $X_t$ depends on $G[N[X_t]]$ cannot be generalized much more:
If we start from \textsc{Maximum Independent Set} on cubic graphs and subdivide each edge four times, we obtain graphs where \textsc{Maximum Independent Set} is \NP-hard, but that admit tree decompositions that contain one large bag $X_t$ such that every connected component of $G[N[X_t]]$ is a $3$-vertex path, while for all other bags $X_t$ their closed neighborhood $N[X_t]$ has bounded size.

Further, we remind the reader that our main result is computationally tight. In particular, in \Cref{thm:inappr}, we proved
that it is unlikely that there is a $g(k)$-approximation algorithm for the tree-independence number with running time $f(k)n^{o(k)}$, for any computable function~$g$. This shows that the $n^{\Omega(k)}$-factor in the running time is unavoidable up to some reasonable complexity assumptions.
For exact computation of the tree-independence number, we proved in  \Cref{the:main_nphard} that it is \NP-complete to decide whether $\tin(G)\leq k$ for every constant $k\geq 4$.
Since $\tin(G)=1$ if and only if $G$ is a chordal graph, \Cref{the:main_nphard} leads to the question about the complexity of deciding whether the tree-independence number is at most $k$ for $k=2$ and~$3$. In \Cref{thm:separator},
we demonstrated that for every fixed integer $k\ge 3$, deciding if two given vertices of a graph can be separated by removing a set of vertices that induces a graph with independence number at most $k$ is \NP-complete. This result
indicates that  it may be already \NP-complete to decide whether $\tin(G)\leq 3$. We hesitate to state any conjecture for the case~$k=2$.

The final question is about the place of computing the tree-independence number in the polynomial hierarchy.
 For a fixed $k$, deciding whether $\tin(G)\leq k$ is in \NP.
However, when $k$ is a part of the input,  the problem is naturally placed in the class $\Sigma^P_2$ on the second level of the polynomial hierarchy.
Is the problem $\Sigma^P_2$-complete?

\section*{Acknowledgments}

We thank the anonymous reviewers for their thoughtful feedback and suggestions, which helped improve the presentation of this work. 
We are especially grateful to one of the reviewers for pointing out an issue in the proof of \Cref{thm:separator} and suggesting a correction.

\end{document}